\documentclass[acmsmall,nonacm]{acmart}

\usepackage{bm} 

\usepackage{hyperref}
\usepackage[colorinlistoftodos]{todonotes}

\usepackage{xspace}
\usepackage{algorithmicx}
\usepackage{algorithm}
\usepackage{algpseudocode}

\usepackage[shortlabels]{enumitem}

\usepackage{tikz}
\usetikzlibrary{calc}
\usetikzlibrary{shapes,decorations.markings,arrows.meta,fit}

\usepackage[normalem]{ulem}

\usepackage{thmtools, thm-restate}

\usepackage{microtype}

\usepackage{relsize}

\newcommand{\rai}{RelationalAI}

\newcommand{\setof}[2]{\{{#1}\mid{#2}\}}        
\newcommand{\dom}{\textsf{Dom}}

\newcommand{\polylog}{\text{\sf polylog}}

\newcommand{\argmin}{\mathrm{argmin}}

\newcommand{\calA}{\mathcal A}

\newcommand{\calI}{\mathcal I}

\newcommand{\calS}{\mathcal S}

\newcommand{\calD}{\mathcal D}

\newcommand{\calT}{\mathcal T}

\newtheorem{thm}{Theorem}[section]

\theoremstyle{definition}              

\newtheorem{claim}{Claim}

\newtheorem{remark}[thm]{Remark}

\newcommand{\defeq}{\stackrel{\text{def}}{=}}

\newcommand{\R}{\mathbb R} 

\newcommand{\cd}{\text{ :- }}

\newcommand{\ceil}[1]{\lceil{#1}\rceil}

\newcommand{\out}{\text{\sf OUT}}

\newcommand{\atoms}{\text{\sf atoms}}
\newcommand{\nodes}{\text{\sf nodes}}
\newcommand{\bs}{\text{\sf BS}}

\newcommand{\true}{\text{\sf true}}

\newcommand{\primalalgo}{\textsf{Jaguar}\xspace}
\newcommand{\calibrate}{\textsf{Calibrate}\xspace}

\newcommand{\heavylightpartition}{\textsf{heavy-light-partition}\xspace}
\newcommand{\equaldegreepartition}{\textsf{equal-degree-partition}\xspace}

\newcommand{\subw}{\textsf{subw}}
\newcommand{\entw}{\textsf{entw}}

\newcommand{\revision}[1]{#1}

\newcommand{\ov}{\overline}
\newcommand{\tl}{\tilde}

\newcommand{\scq}{{\sf{\#CQ}}\xspace}
\newcommand{\scqs}{{\sf{\#CQs}}\xspace}

\allowdisplaybreaks

\begin{document}

\title{Jaguar: A Primal Algorithm for Conjunctive Query Evaluation in Submodular-Width Time}

\author{Mahmoud Abo Khamis}
\orcid{0000-0003-3894-6494}
\email{mahmoudabo@gmail.com}
\affiliation{
  \institution{\rai{}}
  \city{Berkeley}
  \state{CA}
  \country{USA}
}

\author{Hubie Chen}
\orcid{0009-0005-4025-8086}
\email{hubie.chen@kcl.ac.uk}
\affiliation{
  \institution{King's College London}
  \city{London}
  \country{UK}
}

\begin{abstract}
    The \emph{submodular width} is a complexity measure of conjunctive queries, which assigns a nonnegative real number
$\subw(Q)$ to each conjunctive query $Q$.
An existing algorithm, the PANDA algorithm, performs conjunctive query evaluation in polynomial time where the exponent is essentially $\subw(Q)$.  Formally, for every Boolean conjunctive query $Q$, the PANDA algorithm performs conjunctive query evaluation for $Q$ in time $O(N^{\subw(Q)} \cdot \polylog(N))$, where $N$ denotes the input database size; moreover, there is complexity-theoretic evidence that, for a number of Boolean conjunctive queries, no exponent strictly below $\subw(Q)$ can be achieved by combinatorial algorithms.
On the outer level, the submodular width of a conjunctive query $Q$ can be described as the optimum of a maximization problem over certain polymatroids; polymatroids are set functions on the variables of $Q$ that satisfy Shannon inequalities.  The PANDA algorithm in a sense works in the dual space of this maximization problem; it makes use of information theory, manipulates Shannon inequalities, and 
transforms a conjunctive query into a set of disjunctive datalog rules which are individually solved.

In this article, we introduce a new algorithm for conjunctive query evaluation which achieves, for each Boolean conjunctive query $Q$ and for all $\epsilon > 0$, a running time of
$O(N^{\subw(Q)+\epsilon})$.  This new algorithm's
description and analysis are, in our view, significantly simpler than those of PANDA.  We refer to it as a \emph{primal} algorithm as it operates in the primal space of the described maximization problem, by throughout maintaining a feasible primal solution, namely, a polymatroid.  Indeed, this algorithm deals directly with the input conjunctive query and adaptively computes a sequence of joins, in a guided fashion, so that the cost of these join computations is bounded.  
Additionally, this algorithm can achieve the stated runtime for the generalization of submodular width incorporating degree constraints.
We dub our algorithm \emph{\primalalgo}, as it is a join-adaptive guided algorithm.

\end{abstract}



\maketitle

\section{Introduction}
\label{sec:intro}

\emph{Conjunctive query evaluation}, the computational problem of evaluating a conjunctive query on a relational database, is a fundamental problem that appears in many guises all throughout computer science, for example, in database management, graph algorithms, logic, constraint satisfaction, and graphical model inference~\cite{DBLP:conf/pods/KhamisNR16,DBLP:conf/pods/GottlobGLS16,GroheMarx14-fractional-edge-covers,DBLP:books/aw/AbiteboulHV95}.
In the database context, conjunctive queries are considered to be the most commonly occurring queries and are the most heavily studied form of database queries~\cite{DBLP:books/aw/AbiteboulHV95}.


With the motivation of understanding which classes of Boolean conjunctive queries allow for tractable query evaluation, D. Marx introduced a measure of Boolean conjunctive queries called the \emph{submodular width}~\cite{DBLP:journals/jacm/Marx13}.  Let us assume for now that all conjunctive queries under discussion are Boolean.
In an established parameterized complexity model, Marx proved a dichotomy theorem~\cite{DBLP:journals/jacm/Marx13}: on a class of conjunctive queries having \emph{bounded submodular width}---that is, for which there exists a constant bounding the submodular width of all queries in the class---query evaluation is tractable; moreover, on any class of conjunctive queries not having bounded submodular width, query evaluation is intractable, under a complexity-theoretic assumption.  This dichotomy theorem thus comprehensively describes every class of conjunctive queries as either tractable or intractable, with respect to query evaluation, and reveals that the submodular width serves as the critical parameter describing the location of the boundary between these two behaviors.

While Marx's dichotomy theorem describes the tractability status of every class of conjunctive queries, it left open a more fine-grained question regarding the data complexity of conjunctive queries.  For every Boolean conjunctive query $Q$, it is known that conjunctive query evaluation on $Q$ can be performed in polynomial time: there is a constant $\alpha$ such that conjunctive query evaluation on $Q$ can be performed in $O(N^\alpha)$ time; here, $N$ denotes the input database size.  However, there is no general description of how to determine, for a conjunctive query $Q$, the lowest exponent $\alpha$ enjoying this property.  Nonetheless, there are some general time complexity upper bounds known, in this vein, which are described in terms of the submodular width.  Let $\subw(Q)$ denote the submodular width of a Boolean conjunctive query~$Q$.

\begin{itemize}
\item
Abo Khamis, Ngo, and Suciu~\cite{DBLP:conf/pods/Khamis0S17,theoretics:13722} presented an information-theoretic algorithm called \emph{PANDA} that, for any Boolean conjunctive query $Q$, performs conjunctive query evaluation on $Q$
within time $O(N^{\subw(Q)} \cdot \polylog(N))$.

\item
Drawing upon Marx's algorithm and analysis~\cite{DBLP:journals/jacm/Marx13},\footnote{It is worth remarking here that the running time claimed by Marx~\cite{DBLP:journals/jacm/Marx13} for his original algorithm was $N^{O(\subw(Q))}$ times a multiplicative overhead.}
 Berkholz and Schweikardt~\cite{berkholz_et_al:LIPIcs.MFCS.2019.58} presented an algorithm, which we call the Marx-Berkholz-Schweikardt algorithm.  For any Boolean conjunctive query $Q$ and any $\epsilon > 0$, this algorithm performs query evaluation on~$Q$ within time $O(N^{2 \cdot \subw(Q) + \epsilon})$.
\end{itemize}
That is, the first algorithm essentially achieves runtime with exponent $\subw(Q)$, while the second algorithm essentially achieves runtime with exponent of $2 \cdot \subw(Q)$.
Let us mention here that there are techniques for lower bounding the fine-grained time complexity of evaluating conjunctive queries.  In particular, Fan, Koutris, and Zhao~\cite{fan_et_al:LIPIcs.ICALP.2023.127} presented such techniques and---among other results---proved that, for a number of different queries, no combinatorial algorithm can achieve a running time of the form $O(N^{\subw(Q)-\epsilon})$, with $\epsilon > 0$, unless an established fine-grained complexity-theoretic assumption fails.

The two algorithms named above, in addition to differing in terms of the time upper bounds that are demonstrated, differ in two aspects which we wish to highlight.
\begin{itemize}
\item  First, the PANDA algorithm was shown to perform query evaluation not just within the stated time bound
$O(N^{\subw(Q)} \cdot \polylog(N))$, but achieves a running time of similar form where the submodular width is replaced with a natural generalization of the submodular width that incorporates information about the degrees of an input database.  Such information is formalized using a notion which is called \emph{degree constraints} and also, in the context of conjunctive query evaluation, is simply called a \emph{statistics}.  A \emph{statistics} is a pair $(\Delta, \bm n)$ where, speaking on a high level, each element of $\Delta$ is a term consisting of a pair of variable subsets, and where $\bm n$ maps each term in $\Delta$ to a number.
The mentioned generalization of the submodular width~\cite{theoretics:13722}
maps each query $Q$ and statistics $(\Delta, \bm n)$
to a value $\subw(Q, \Delta, \bm n)$, and PANDA can achieve a running time of
$O(N^{\subw(Q, \Delta, \bm n)} \cdot \polylog(N))$ when evaluating
a query $Q$ on databases satisfying the statistics $(\Delta, \bm n)$.
For any query $Q$ and any statistics $(\Delta, \bm n)$,
it holds that $\subw(Q, \Delta, \bm n) \leq \subw(Q)$;
it is thus transparent to see that incorporating statistics can only result in an improved running time.
We are not aware of any work showing that an algorithm along the lines of the Marx-Berkholz-Schweikardt algorithm can utilize this information.

\item Second, there is---in our view---a notable difference
both in the objects that these algorithms and their analyses
reason about, and also in the tools used by the analyses.
To describe this difference, we begin by reflecting on the definition of submodular width.  The submodular width is defined as the optimum of a maximization problem over the space of certain polymatroids,
where the goal is to find the polymatroid maximizing a certain objective.
The dual problem is a minimization problem over the space of Shannon inequalities, where the goal is
to infer a Shannon inequality yielding the tightest upper bound on the submodular width.
It can be shown by LP duality that these two problems have the same optimal value~\cite{theoretics:13722}.
\vspace{2pt}

The PANDA algorithm operates in the dual space, that is, the space of Shannon inequalities;
it always maintains a Shannon inequality that upper bounds the submodular width.
In contrast, the Marx-Berkholz-Schweikardt algorithm deals directly with the conjunctive query under evaluation; it uses recursion, and for any run of the algorithm, it is shown that at each leaf of the recursive tree, there is a polymatroid describing the data at that leaf.
We view PANDA as a \emph{dual} algorithm and the Marx-Berkholz-Schweikardt algorithm as a \emph{primal} algorithm.  Although we do not
give formal definitions of \emph{dual} and \emph{primal} algorithms,
we believe that the distinction is a meaningful one; certainly, information theory is not explicitly employed by the Marx-Berkholz-Schweikardt algorithm nor its analysis.

\end{itemize}

\paragraph*{\bf Contributions.}
We introduce a new algorithm for conjunctive query evaluation called \emph{\primalalgo}.
Our algorithm has the following features.

\begin{itemize}

\item It is fast. It achieves, for every query $Q$ and value $\epsilon > 0$, a running time of
$O(N^{\subw(Q) + \epsilon})$.  Its running time is thus similar to that of PANDA.

\item It is general.  It can make use of degree information; for every query $Q$, any statistics $(\Delta, \bm n)$, and any value $\epsilon > 0$, it can achieve a running time of
$O(N^{\subw(Q, \Delta, \bm n) + \epsilon})$
when evaluating $Q$ on databases satisfying the statistics $(\Delta, \bm n)$.

\item It and its analysis are relatively simple.
With respect to the above distinction between dual and primal algorithms, we view \primalalgo as a primal algorithm.  Throughout an execution, \primalalgo maintains a set function $\bm g$ that is defined on all variable subsets and that describes the data; this set function---in particular, the points at which it violates submodularity---is used to guide the process of partitioning the data and computing joins.  Our analysis shows and crucially uses the fact that this set function induces a polymatroid that---in a certain precise sense---lower bounds the submodular width.
\vspace{2pt}

We believe that \primalalgo and its analysis directly shed light on the role of submodularity in conjunctive query evaluation.
We view \primalalgo and its analysis as simpler than that of PANDA; note that its study requires no explicit use of information theory.  At the same time, we view \primalalgo and its analysis as simpler and more direct than that of the Marx-Berkholz-Schweikardt algorithm, which makes use of two notions of \emph{consistency} and, throughout its execution, generically searches for pairs of subsets of the variables that violate a condition called \emph{$\epsilon$-uniformity}; when a violation is found, a general partitioning procedure is invoked.

\end{itemize}

To sum up, we believe that the introduction of \primalalgo and its analysis is valuable, since it (1) essentially achieves the runtime exponent $\subw(Q, \Delta, \bm n)$ that can be essentially achieved by PANDA, but is not known to be achievable by the Marx-Berkholz-Schweikardt algorithm, which essentially achieves an exponent of $2 \cdot \subw(Q)$, and (2) accomplishes this via a simple algorithm having an analysis that is, in our view, quite direct and deals very openly and clearly with the notion of submodularity.

Let us remark that our main result on \primalalgo's performance addresses not just Boolean conjunctive queries, but also addresses conjunctive queries that may have free variables.  We work with a suitable and established generalization of submodular width to arbitrary conjunctive queries, and show that \primalalgo achieves a runtime of
$O(N^{\subw(Q, \Delta, \bm n)+\epsilon}+ \out)$, where $\out$ denotes the output size.
Note that each of the PANDA and the Marx-Berkholz-Schweikardt algorithms, on general conjunctive queries, can similarly achieve the running time described for them, plus the output size.

\paragraph*{\bf Related work.}  The introduction of the submodular width
followed a long and substantial body of work defining and studying complexity measures of conjunctive queries, including hypertree width~\cite{GottlobLeoneScarcello02-hypertreedecomposisions,GottlobLeoneScarcello03-characterizations}, generalized hypertree width~\cite{ChenDalmau05-hypertree,AdlerGottlobGrohe07-hypergraphinvariants}, and fractional hypertree width~\cite{Marx10-approximating-fhw,GroheMarx14-fractional-edge-covers,GottlobLPR21-analysis-decompositions}; here, we have offered a sampling of references.  For each conjunctive query $Q$, all of these measures provide upper bounds on the discussed exponent describing the running time of query evaluation on $Q$.
Abo Khamis, Hu, and Suciu~\cite{DBLP:journals/pacmmod/KhamisHS25} introduced an algebraic extension of the submodular width whose corresponding algorithm can be seen as combining techniques used by PANDA
with fast matrix multiplication.

As mentioned, Fan, Koutris, and Zhao~\cite{fan_et_al:LIPIcs.ICALP.2023.127}
presented techniques for lower bounding this discussed exponent and, for a number of queries, demonstrated (essentially) that the submodular width serves as a lower bound for combinatorial algorithms; these results thus in a sense show that, on these queries, PANDA and \primalalgo yield optimal running times---as combinatorial algorithms.
Related lower bound technology that also applies to non-combinatorial algorithms was developed by Bringmann and Gorbachev~\cite{10.1145/3717823.3718141}, who studied queries~$Q$ where $\subw(Q)<2$ and all relations are binary, and they focused on two variants of these queries:
\begin{itemize}
    \item {\em Aggregate} queries over the {\em tropical semiring}, i.e., $(min, +)$.
    \item {\em Full} conjunctive queries, i.e., conjunctive queries {\em without projections}.
\end{itemize}
Bringmann and Gorbachev proved that in the regime where $\subw(Q)<2$ and all relations are binary,
the submodular width serves as a lower bound for both query variants {\em even for non-combinatorial algorithms}, under well-accepted fine-grained complexity conjectures.
In particular, for both aggregate queries over the tropical semiring and for full conjunctive queries, using fast matrix multiplication (FMM) cannot help lower the runtime.
The intuition is that for the tropical semiring, FMM does not immediately apply because the tropical semiring does not have additive inverses, and for full conjunctive queries, FMM cannot does not help because
it can only compute a join {\em followed by a projection}.\footnote{Given two relations $R(X,Y)$ and $S(Y,Z)$, FMM allows us to compute the join $R(X,Y) \Join S(Y,Z)$ followed by a projection $\pi_{X,Z}(R(X,Y) \Join S(Y,Z))$; however, it does not immediately allow us to list the $Y$-values, hence does not immediately compute the {\em full} join.}
See~\cite{10.1145/3717823.3718141} for the formal lower bound arguments.

Fan, Koutris, and Zhao~\cite{10.1145/3651588} presented a different form of lower bounds, namely, circuit lower bounds, relative to a measure related to the submodular width; we discuss this work further in Section~\ref{sec:conclusion}.

%


%
%

\revision{\paragraph*{\bf Organization.}
Preliminaries and notation are given in Section~\ref{sec:notation}.
Section~\ref{sec:truncation-lemma} explains a key technical lemma that will be at the heart of the new \primalalgo algorithm.
In Section~\ref{sec:example}, we present a warmup example that demonstrates the main ideas and intuition behind the algorithm.
In Section~\ref{sec:algorithm}, we explain the general \primalalgo algorithm in detail.
We prove its correctness and analyze its runtime in Section~\ref{sec:analysis}.
Finally, in Section~\ref{sec:conclusion}, we conclude with a recap and some open problems.}

\section{Preliminaries}
\label{sec:notation}

Given a function $f:A\to B$, a value $a \in A$, and a value $b \in B$,
we define $f[a \to b]$ to be a new function $f':A\to B$ such that
\begin{align}
    f'(x) &\defeq \begin{cases}
        b, &\text{if } x = a\\
        f(x), &\text{otherwise}
    \end{cases}
    \label{eq:function-update}
\end{align}

Throughout the paper,
we {\em fix} a finite set of variables $\bm V$ where each variable $X \in \bm V$ takes values from
an infinite domain $\dom$.
We use a capital letter $X$ to denote a variable and a small letter $x$ to denote a value of the corresponding variable $X$.
We use a bold capital letter $\bm X$ to denote a set of variables and a bold small letter $\bm x$ to denote a tuple of values for the variables in $\bm X$.
We use $\dom^{\bm X}$ to denote the set of all possible tuples of values $\bm x$ for the variables in $\bm X$.
\revision{Given a tuple $\bm x \in \dom^{\bm X}$ and a subset $\bm Y\subseteq \bm X$,
we use $\pi_{\bm Y}(\bm x)$ to denote the {\em projection of $\bm x$ onto $\bm Y$}, which is the tuple of values for the variables in $\bm Y$ that are consistent with $\bm x$.}

\paragraph*{\bf Database Instances}
A database instance $\calD$ is a finite set of relation instances $\{R_1^\calD(\bm X_1),$ $\ldots,$$ R_m^\calD(\bm X_m)\}$,
where each relation instance $R_i^\calD(\bm X_i) \subseteq \dom^{\bm X_i}$ is a finite set of tuples over the variables $\bm X_i \subseteq \bm V$.
When the database instance $\calD$ is clear from the context, we drop the superscript and write $R_i(\bm X_i)$ instead of $R_i^\calD(\bm X_i)$.
The {\em size} of a database instance $\calD$, denoted $|\calD|$, is defined as the total number of tuples in all relation instances in $\calD$, i.e., $|\calD| \defeq \sum_{i=1}^m |R_i(\bm X_i)|$.
A database instance $\calD$ is called {\em signature-unique} if, for every $\bm X \subseteq \bm V$, $\calD$ contains at most one relation instance $R(\bm X)$.
\revision{Note that according to this definition, a signature-unique database instance $\calD$
can still contain two relation instances $R_1(\bm X_1)$ and $R_2(\bm X_2)$ where $\bm X_1\subset \bm X_2$.}
Given a signature-unique database instance $\calD$ and a relation instance $R(\bm X)$ where $\bm X \subseteq \bm V$, we define {\em $\calD$ augmented with $R(\bm X)$}, denoted $\calD\uplus \{R(\bm X)\}$, to be a new database instance 
that results from adding $R(\bm X)$ to $\calD$, if no $R'(\bm X)$ is already in $\calD$,
or, otherwise, by replacing the existing $R'(\bm X)$ with its intersection with $R(\bm X)$:
\begin{align}
    \calD\uplus \{R(\bm X)\} &\defeq \begin{cases}
        \calD \cup \{R(\bm X)\}, &\text{if there is no } R'(\bm X) \in \calD\\
        (\calD \setminus \{R'(\bm X)\}) \cup \{R(\bm X) \cap R'(\bm X)\}, &\text{if there is } R'(\bm X) \in \calD
    \end{cases}
    \label{eq:db-augmentation}
\end{align}
Note that $\calD \uplus \{R(\bm X)\}$ is also signature-unique.

\revision{We use the standard notions of {\em projection} and {\em selection} of relations.
In particular, given a relation instance $R(\bm Z)$ and a subset $\bm X\subseteq \bm Z$,
we define the {\em projection of $R(\bm Z)$ onto $\bm X$}, denoted $\pi_{\bm X}(R)$, as the relation instance over $\bm X$ that contains the projections of all tuples in $R$ onto $\bm X$.
Moreover, given a tuple $\bm x\in \dom^{\bm X}$, we use $\sigma_{\bm X=\bm x}(R)$
to denote the {\em selection of $R$ with $\bm X = \bm x$}, which is the relation instance over $\bm Z$ that contains all tuples $\bm z$ in $R$ that are consistent with $\bm x$, i.e., where $\pi_{\bm X}(\bm z) = \bm x$.}
Given a relation instance $R(\bm Z)$, two sets of variables $\bm X, \bm Y \subseteq \bm Z$,
and a tuple $\bm x$ for $\bm X$,
we define the {\em degree} of $\bm Y$ given $\bm X = \bm x$ in $R$, denoted $\deg_R(\bm Y|\bm X = \bm x)$, and the {\em degree} of $\bm Y$ given $\bm X$ in $R$, denoted $\deg_R(\bm Y|\bm X)$, as follows, respectively:
\begin{align*}
    \deg_R(\bm Y|\bm X=\bm x) &\defeq
        |\pi_{\bm Y}(\sigma_{\bm X =\bm x}(R))|\\
    \deg_R(\bm Y|\bm X) &\defeq \max_{\bm x\in \dom^{\bm X}} \deg_R(\bm Y|\bm X=\bm x)
\end{align*}


\paragraph*{\bf Conjunctive Queries}
A conjunctive query (CQ) $Q$ is specified with a set of variables $\bm V$, a set of {\em free variables} $\bm F$, and a set of atoms $\atoms(Q)$, and has the form:
\begin{align}
    Q(\bm F) &\cd  \bigwedge_{R(\bm X)\in \atoms(Q)} R(\bm X), \label{eq:cq}
\end{align}
where $\bm V = \bigcup_{R(\bm X)\in\atoms(Q)} \bm X$.
The query $Q$ is called {\em Boolean} if $\bm F = \emptyset$ and {\em full} if $\bm F = \bm V$.
A database instance $\calD$ is said to {\em include the schema} of $Q$ if, for every atom $R(\bm X) \in \atoms(Q)$, $\calD$ contains a corresponding relation instance $R^\calD(\bm X)$.
A database instance $\calD$ is said to {\em match the schema} of~$Q$
if it includes the schema of $Q$ {\em and} contains no other relation instances.
We adopt the standard semantics for evaluating a CQ $Q$ on a database instance $\calD$ that includes the schema of $Q$, and we use $Q(\calD)$ to denote the evaluation result.
Throughout the paper, we use {\em data complexity}, i.e., we assume that the query $Q$ is fixed, hence its size is a constant, and we measure the complexity of evaluating $Q$ only in terms of the size, $N$, of the input database instance $\calD$.

\paragraph*{\bf Statistics and Polymatroids.}
Using definitions and notation from~\cite{theoretics:13722}, we define a collection of {\em statistics} over a variable set $\bm V$ to be a pair $(\Delta, \bm n)$, where:
\begin{itemize}
    \item $\Delta$ is a set of {\em conditional terms} of the form $\bm Y|\bm X$ where $\bm X, \bm Y\subseteq \bm V$.
    \item $\bm n:\Delta \rightarrow\R_+$ maps each term $\bm Y|\bm X$ in $\Delta$ to a non-negative real number $n_{\bm Y|\bm X}$.
\end{itemize}
Let $\calD$ be a database instance of size $N$.
The instance $\calD$ is said to {\em satisfy} the statistics $(\Delta, \bm n)$,
denoted $\calD\models (\Delta,\bm n)$,
if for every conditional term $(\bm Y|\bm X)\in \Delta$,\;
$\calD$ contains a relation $R(\bm Z)$,  called the \emph{guard} of $\bm Y|\bm X$,
that has variables $\bm Z\supseteq \bm X \cup\bm Y$
and satisfies:
\begin{align}
    \log_N\deg_R(\bm Y|\bm X) &\leq n_{\bm Y|\bm X}
    \label{eq:DC-satisfaction}
\end{align}
Throughout the paper, {\em all} logarithms are base $N$ where $N$ is the size of the input database instance~$\calD$. This is in contrast to~\cite{theoretics:13722}, where logarithms are base $2$.

A set function $\bm h:2^{\bm V}\rightarrow \R_+$ is called a {\em polymatroid} if it satisfies {\em Shannon inequalities}:
\begin{align}
    h(\emptyset) = 0&\label{eq:normalization} \\
    h(\bm X) \leq h(\bm Y), &\quad\forall \bm X \subseteq \bm Y \subseteq \bm V, &\text{(Monotonicity)}\label{eq:monotonicity}\\
    h(\bm X) + h(\bm Y) \geq h(\bm X \cap \bm Y) + h(\bm X \cup \bm Y), &\quad\forall \bm X, \bm Y \subseteq \bm V, &\text{(Submodularity)}\label{eq:submodularity}
\end{align}
A set function $\bm h:2^{\bm V}\rightarrow \R_+\cup\{\infty\}$ is said to {\em satisfy} given statistics $(\Delta, \bm n)$,
denoted $\bm h\models(\Delta, \bm n)$,
if for every conditional term $\bm Y|\bm X\in \Delta$, we have $h(\bm Y|\bm X)\defeq h(\bm Y\cup \bm X)-h(\bm X)\leq n_{\bm Y|\bm X}$.
(In general, we deal with arithmetic expressions and comparisons involving infinity in the standard way.  So here, we understand that this condition fails when
$h(\bm Y\cup \bm X) = \infty$, and that it holds when
$h(\bm Y\cup \bm X) \neq \infty = h(\bm X)$.
)

\paragraph*{\bf Tree decompositions}
A {\em tree decomposition} of $Q$ is
a pair $(T,\chi)$, where $T$ is a tree and $\chi: \nodes(T) \rightarrow 2^{\bm V}$ is a
map from the nodes of $T$ to subsets of $\bm V$ that satisfies the following properties: for every atom $R(\bm X)\in\atoms(Q)$,
there is a node $t \in \nodes(T)$ such that $\bm X \subseteq \chi(t)$; and, for every variable
$X \in \bm V$, the set $\setof{t}{X \in \chi(t)}$ forms a connected sub-tree of $T$. Each
set $\chi(t)$ is called a {\em bag} of the tree decomposition, and we will assume w.l.o.g.
that the bags are distinct, i.e.~$\chi(t)\neq\chi(t')$ when $t\neq t'$ are nodes in $T$.
A {\em free-connex} tree decomposition of~$Q$ is a tree decomposition of $Q$ with the
additional property that there is a connected subtree $T'$ of $T$
that contains all the free variables and no other variables~\cite{DBLP:conf/csl/BaganDG07,10.1145/3426865}, i.e.,
$\bigcup_{t \in \nodes(T')} \chi(t) = \bm F$.
Let $\calT(Q)$ be the set of {\em free-connex} tree decompositions of $Q$.
(If the query is Boolean or full, i.e., if $\bm F = \emptyset$ or $\bm F = \bm V$,
then $\calT(Q)$ is the set of {\em all} tree decompositions of $Q$.
Readers only interested in Boolean or full conjunctive queries can thus ignore the free-connex notion.)

\paragraph*{\bf The Submodular Width}
The submodular width was originally introduced by Daniel Marx~\cite{DBLP:journals/jacm/Marx13} and was later generalized in~\cite{theoretics:13722}.
We adopt the generalized definition from~\cite{theoretics:13722}.
In particular, given a query $Q$ and statistics $(\Delta, \bm n)$, we define the {\em submodular width} of $Q$ w.r.t.~the statistics $(\Delta, \bm n)$ as:
\begin{align}
\subw(Q,\Delta, \bm n) &\quad\defeq\quad
    \max_{\substack{\text{Polymatroid $\bm h$}\\\bm h\models(\Delta, \bm n)}}\quad
    \min_{(T,\chi) \in \calT(Q)}\quad
    \max_{t \in\nodes(T)}\quad
    h(\chi(t))
    \label{eq:subw}
\end{align}
The above definition generalizes the original definition of Daniel Marx~\cite{DBLP:journals/jacm/Marx13} in two aspects:
\begin{itemize}
    \item It accounts for arbitrary statistics $(\Delta, \bm n)$, as opposed to the special case of ``edge-domination'' constraints\footnote{Edge-domination constraints are restricted to the form $h(\bm X)\leq 1$ for every atom $R(\bm X)$ in the query.} from~\cite{DBLP:journals/jacm/Marx13}.
    \item It handles arbitrary conjunctive queries (that are not necessarily Boolean or full)
by using the notion of {\em free-connex} tree decompositions.
\end{itemize}
\begin{remark}
    Since in Eq.~\eqref{eq:DC-satisfaction}, we use logarithms base $N$ where $N$ is the size of the input database instance $\calD$,
    our runtimes later will be stated in terms of $N^{\subw(Q,\Delta, \bm n)}$.
    This is more consistent with the original definition of the submodular width~\cite{DBLP:journals/jacm/Marx13}.
    However, it deviates from~\cite{theoretics:13722} where logarithms were base $2$, hence runtimes were stated in terms of $2^{\subw(Q,\Delta, \bm n)}$.
    \label{rmk:log-base}
\end{remark}

Given a conjunctive query $Q$ with variables $\bm V$ and tree decompositions $\calT(Q)$,
a {\em bag selector} $\calS$ is a subset of $2^{\bm V}$ that
consists of a bag $\chi(t)$ out of every tree decomposition $(T, \chi)\in\calT(Q)$.
Let $\bs(Q)$ denote the set of all bag selectors of $Q$.
By distributing the min over the inner max in Eq.~\eqref{eq:subw}, we can rewrite the submodular width as follows:
\begin{align}
\subw(Q,\Delta, \bm n) &\quad=\quad
    \max_{\substack{\text{Polymatroid $\bm h$}\\\bm h\models(\Delta, \bm n)}}\quad
    \max_{\calS \in \bs(Q)} \quad
    \min_{\bm Z \in \calS}\quad
    h(\bm Z)
    \label{eq:subw:bs}\\
    &\quad=\quad
    \max_{\calS \in \bs(Q)} \quad
    \max_{\substack{\text{Polymatroid $\bm h$}\\\bm h\models(\Delta, \bm n)}}\quad
    \min_{\bm Z \in \calS}\quad
    h(\bm Z)
    \nonumber
\end{align}

\paragraph*{\bf Constant-delay enumeration}
Given a query $Q$ and a database instance $\calD$ that includes the schema of $Q$,
{\em constant-delay enumeration (CDE)} refers to enumerating tuples in the output $Q(\calD)$ one by one without repetitions while keeping the {\em delay} constant. The delay is the maximum of three times: the time to output the first tuple, the time between any two consecutive tuples, and the time after outputting the last tuple until the end of enumeration. See~\cite{10.1145/2783888.2783894} for a survey.
\section{Truncation Lemma}
\label{sec:truncation-lemma}

Before describing the $\primalalgo$ algorithm, we present a key technical lemma that is at the heart of this algorithm's analysis.
In essence, this lemma converts a monotone
set function $\bm g$ to a polymatroid $\bm h$.
The polymatroid $\bm h$ is constructed
 by taking the lowest value $c$
at which $\bm g$ violates submodularity,
and truncating the function $\bm g$ 
above $c$, so that $c$ becomes the maximum value after truncation.
The truncated $\bm h$ no longer violates submodularity, hence is a polymatroid.
But now, we can utilize $\bm h$ in the context of the submodular width as follows.
By construction, the polymatroid~$\bm h$ has a maximum value of $c$, and, in turn, the submodular width is defined as a maximum
over polymatroids.
Hence, the lemma shows that $c$ is upper bounded by the submodular width, under a certain condition.

Given a query $Q$ over variable set $\bm V$ and a tree decomposition $(T,\chi)\in\calT(Q)$,
a set $\calI\subseteq 2^{\bm V}$ is said to \emph{cover} the tree decomposition $(T,\chi)$
if every bag $\chi(t)$ of $(T,\chi)$ is in $\calI$.
\begin{restatable}[Truncation lemma]{lemma}{TruncationLemma}
    \label{lmm:primal-feasible}
    Let $\bm g:2^{\bm V}\rightarrow \R_+\cup\{\infty\}$ be a monotone set function satisfying $g(\emptyset)=0$.
    Define 
    the following:\footnote{
    Note that $f$ can be undefined, for example when
     $g(\bm X) = g(\bm X \cap \bm Y) = \infty$.
     We follow the convention that the minimum over an empty set is $\infty$, hence $c \in \R_+\cup\{\infty\}$.}
    \begin{align}
        f(\bm X, \bm Y) &\quad\defeq\quad g(\bm X) + g(\bm Y) - g(\bm X \cap \bm Y),&\forall \bm X, \bm Y\subseteq \bm V\label{eq:f(X,Y)}\\
        c &\quad\defeq\quad \min_{\substack{\bm X, \bm Y\subseteq \bm V\\g(\bm X \cup \bm Y)>f(\bm X, \bm Y)}} f(\bm X, \bm Y)\label{eq:c}\\
        h(\bm X) &\quad\defeq\quad \min(g(\bm X), c), &\forall \bm X\subseteq \bm V\label{eq:h}\\
        \calI &\quad\defeq\quad \setof{\bm X \subseteq \bm V}{g(\bm X) \leq c}\label{eq:I}
    \end{align}
    Then, the following claims hold:
    \begin{enumerate}
        \item The function $\bm h$ is a polymatroid that satisfies $\bm h \leq \bm g$ and $\bm h \leq c$.\label{lmm:primal-feasible:1}
        \item Let $Q$ be a CQ with variable set $\bm V$ and let $(\Delta, \bm n)$
        be a collection of statistics over $\bm V$ such that $\bm g\models(\Delta, \bm n)$.
        Then, $\bm h\models(\Delta, \bm n)$.
        In addition, if $\calI$ does not cover any tree decomposition of $Q$,
        then $c \leq \subw(Q,\Delta, \bm n)$.\label{lmm:primal-feasible:2}
    \end{enumerate}
\end{restatable}

\begin{proof}
    First, we prove Claim~(\ref{lmm:primal-feasible:1}).
    It is straightforward to verify that the function $\bm h$ 
    always takes on values in $\R_+$ and that $\bm h$ is upper bounded by both $\bm g$ and $c$.    
    Since $\bm g$ is monotone and satisfies $g(\emptyset)=0$,
    the function $\bm h$ is also monotone and satisfies $h(\emptyset)=0$.
    It remains to show that $\bm h$ is submodular, that is, for every $\bm X, \bm Y\subseteq \bm V$, we have $h(\bm X\cup\bm Y) \leq h(\bm X) + h(\bm Y)- h(\bm X \cap \bm Y)$.
    We recognize the following cases:
    \begin{itemize}
        \item Case 1: $g(\bm X), g(\bm Y) \leq c$.
        This implies that $h(\bm X) = g(\bm X)$, $h(\bm Y) = g(\bm Y)$, and $h(\bm X \cap \bm Y)=g(\bm X \cap \bm Y)$.
        If $g(\bm X \cup \bm Y) \leq f(\bm X,\bm Y)$, then
        $h(\bm X \cup \bm Y) \leq g(\bm X \cup \bm Y) \leq f(\bm X,\bm Y)$ and we are done.
        Otherwise, by definition of $c$, we have $h(\bm X \cup \bm Y) = c \leq f(\bm X, \bm Y)$.
        \item Case 2: $g(\bm X) \leq c$ and $g(\bm Y) > c$.
        This implies that $g(\bm X \cap \bm Y) \leq c$ and $g(\bm X\cup \bm Y) > c$.
        Therefore, $h(\bm X \cup \bm Y) \leq h(\bm X) + h(\bm Y) - h(\bm X \cap \bm Y)$
        holds because $h(\bm X \cup \bm Y) = h(\bm Y) = c$ and monotonicity
        implies that $h(\bm X \cap \bm Y) \leq h(\bm X)$.
        \item Case 3: $g(\bm Y) \leq c$ and $g(\bm X) > c$. Symmetric to Case 2.
        \item Case 4: $g(\bm X), g(\bm Y) > c$. This implies that $g(\bm X \cup \bm Y) > c$
        and
        $h(\bm X \cup \bm Y) = h(\bm X) = h(\bm Y) = c$.
        Since $h(\bm X\cap \bm Y) \leq c$ by definition of $\bm h$, the submodularity holds.
    \end{itemize}
    Now, we prove Claim~(\ref{lmm:primal-feasible:2}).
    We start by showing that $\bm h\models (\Delta, \bm n)$.
    For every $\bm X, \bm Y \subseteq \bm V$, we have
    \begin{align*}
        h(\bm Y|\bm X) &\defeq h(\bm Y\cup \bm X) - h(\bm X)\\
        &=\min(g(\bm X\cup \bm Y), c) - \min(g(\bm X), c)\\
        &\leq g(\bm Y\cup \bm X) - g(\bm X)&\text{(By monotonicity of $\bm g$)}\\
        &\leq n_{\bm Y|\bm X}&\text{(Because $\bm g\models(\Delta, \bm n)$)}
    \end{align*}
    This proves that $\bm h\models(\Delta, \bm n)$.
    Now we prove that $c \leq \subw(Q,\Delta, \bm n)$.
    If $\calI$ does not contain all the bags of any tree decomposition in $\calT(Q)$,
    then the complement of $\calI$, denoted $\ov{\calI}\defeq 2^{\bm V}\setminus \calI$,
    must contain at least one bag of every tree decomposition in $\calT(Q)$.
    Hence, $\ov{\calI}$ must be a superset of some {\em bag selector} $\calS^* \in \bs(Q)$.
    By Eq.~\eqref{eq:subw:bs}, we have:
    \begin{align*}
        \subw(Q,\Delta, \bm n) &=
            \max_{\substack{\text{Polymatroid }\bm h\\\bm h\models(\Delta, \bm n)}}\quad
            \max_{\calS \in \bs(Q)} \quad
            \min_{\bm Z \in \calS}\quad
            h(\bm Z)\\
            &\geq
            \max_{\calS \in \bs(Q)} \quad
            \min_{\bm Z \in \calS}\quad
            h(\bm Z)&\text{(By choosing $\bm h$ from Eq.~\eqref{eq:h})}\\
            &\geq
            \min_{\bm Z \in \calS^*}\quad
            h(\bm Z)&\text{(By choosing $\calS^*$)}\\
            &=c&\text{(because $h(\bm Z) = c$ for every $\bm Z\in \ov{\calI}\supseteq \calS^*$)}
    \end{align*}
\end{proof}
\revision{\section{Warmup Example}
\label{sec:example}

Before presenting the full \primalalgo algorithm in Section~\ref{sec:algorithm}, we start in this section with a warmup example that demonstrates the main ideas behind the algorithm.}
Consider the following CQ $Q_\square$:
\begin{equation}
\label{eq:4cycle}
Q_\square(X, Y, Z, W) \cd R(X,Y), S(Y,Z), T(Z,W), U(W,X).
\end{equation}
The query $Q_\square$ is known to have a {\em fractional hypertree width}~\cite{DBLP:conf/soda/GroheM06} of 2
and a submodular width of 3/2~\cite{DBLP:journals/algorithmica/AlonYZ97,DBLP:conf/pods/Khamis0S17}.
In addition to the trivial tree decomposition that puts all variables in one bag, $Q_\square$ has two free-connex tree decompositions:
\begin{itemize}
    \item $\calT_1$ with bags $\{X,Y,Z\}$ and $\{Z,W,X\}$, and
    \item $\calT_2$ with bags $\{W,X,Y\}$ and $\{Y,Z,W\}$.
\end{itemize}
The query and both tree decompositions are depicted in Figure~\ref{fig:4cycle}.
Let $N$ be an {\em even} number.
Consider the following input database instance $\calD_\square$, where $[N]$ is a shorthand for $\{1, 2, \ldots, N\}$:
\begin{align*}
    R = S = T = U = ([N/2]\times \{1\}) \cup (\{1\}\times [N/2])
\end{align*}
Note that no matter which tree decomposition we choose,
there is always a bag whose size is $\Omega(N^2)$ over this database instance.
For example, the bag $\{X, Y, Z\}$ of $\calT_1$
requires computing the join $R(X,Y) \Join S(Y,Z)$, which has size $\Omega(N^2)$.
This is consistent with the fact that the fractional hypertree width of $Q_\square$ is 2:
Using a single tree decomposition, we cannot beat the $\Omega(N^2)$ runtime in the worst-case.
We show below how the \primalalgo algorithm can achieve a runtime of $O(N)$
on the specific input database instance $\calD_\square$ by utilizing both tree decompositions $\calT_1$ and $\calT_2$,
thus mimicking the algorithm from~\cite{DBLP:journals/algorithmica/AlonYZ97}.
We focus on this database instance for the purpose of illustrating our algorithm; note that the general time bound achieved by our algorithm on the query $Q_\square$ is $O(N^{3/2 + \epsilon})$, as this query's submodular width is $3/2$.

At a high level, the \primalalgo algorithm starts by initializing a monotone set function $\bm g:2^{\bm V}\to\R_+\cup\{\infty\}$ where $\bm V :=\{X, Y, Z, W\}$ in this example
using the sizes of the input relations and their projections on $\log_N$-scale.
In particular, the function $\bm g$ is initialized as follows:
\begin{align*}
    g(XY) &= g(YZ) = g(ZW) = g(WX) = \log_N N = 1,\\
    g(X) &= g(Y) = g(Z) = g(W) = \log_N N = 1,\\
    g(\emptyset) &= \log_N 1 = 0,\\
    g(\bm X) &= \infty \quad\quad\text{ otherwise}.
\end{align*}
Now, we find the ``smallest'' violation of submodularity by $\bm g$.
Specifically, we look for a violated submodularity
$g(\bm X\cup\bm Y) > f(\bm X, \bm Y) \defeq g(\bm X) + g(\bm Y) - g(\bm X \cap \bm Y)$ that minimizes $f(\bm X, \bm Y)$, in the same spirit as in Lemma~\ref{lmm:primal-feasible}.
Breaking ties arbitrarily, we 
could pick $g(XYZ) = \infty > g(XY) + g(YZ) - g(Y) = 1$.
In order to fix this submodularity, we need to lower
the value of $g(XYZ)$ from $\infty$ down to 1 (or less).
To translate this change in the function $\bm g$ into the database world, 
we need to be able compute a relation over $\{X, Y, Z\}$ of size at most $N^1=N$.
Note that $g(Z|Y) \defeq g(YZ) - g(Y) = 0$.
We interpret this as saying that
we can only process a $Y$-value in $S$ if it has at most one matching $Z$-value.
We call such $Y$-values \emph{light}, and the rest of the $Y$-values \emph{heavy}.
Now, the execution of the algorithm branches into two corresponding branches:
\begin{itemize}
    \item On the {\em light}-$Y$ branch, we compute the join $R(X, Y) \Join S(Y, Z)$
    but {\em only} over $Y$-values in $S$ that are {\em light}, i.e., have at most one matching $Z$-value in $S$.
    The size of this join is at most $N$,
    thus allowing us to lower the value of $g(XYZ)$ down to $\log_N N = 1$.
    Thanks to this change, the submodularity is fixed $g(XYZ)=1 \leq g(XY) + g(YZ) - g(Y) = 1$.
    Now we can
    look for the next ``smallest'' violation among the remaining submodularities.
    Breaking ties arbitrarily again, suppose now we pick $g(ZWX) = \infty > g(ZW) + g(WX) - g(W) = 1$.
    Just like before, now we partition the $W$-values into light and heavy based on the number of matching $X$-values in $U(W, X)$, and we create two corresponding branches:
    \begin{itemize}
        \item On the light-$W$ branch, we compute the join $T(Z, W)\Join U(W, X)$ over only the light $W$-values, hence this join has size at most $N$.
        After we compute this join, we now have two relations over $\{X, Y, Z\}$ and $\{Z, W, X\}$ corresponding to the two bags of $\calT_1$.
        We can enumerate the output of $Q_\square$ by running Yannakakis algorithm~\cite{DBLP:conf/vldb/Yannakakis81} over $\calT_1$ in time $O(\out)$,
        where $\out$ is the output size for this particular branch.
        \item On the heavy-$W$ branch, we cannot afford to compute the join $T(Z, W)\Join U(W, X)$ over the heavy $W$-values, thus we are not allowed to lower the value of $g(ZWX)$ down to 1.
        Instead, we take advantage of the fact that the number of heavy $W$-values is very small, namely 1 in this particular database instance $\calD_\square$.
        This allows us to lower the value of $g(W)$ down to $\log_N 1= 0$.
        Note that this change does not really fix the submodularity $g(ZWX) = \infty > g(ZW) + g(WX) - g(W) = 2$.
        Nevertheless, we can now look for the ``smallest'' submodularity violation in the resulting $\bm g$.
        Let's say we pick $g(YZW) = \infty \geq g(YZ) + g(W) = 1$.
        To fix this submodularity, we need to lower the value of $g(YZW)$ down to 1 (or less).
        In the database world, we need to be able to mirror this change by computing a relation over $\{Y, Z, W\}$ of size at most $N^1=N$.
        Luckily, we can compute this directly by expanding tuples in $S(Y, Z)$
        with the single heavy $W$-value.
        Now we look for another submodularity violation, let's say we pick $g(WXY) = \infty > g(W) + g(XY) = 1$.
        Just like before, we need to lower the value of $g(WXY)$ down to 1, which corresponds to computing a relation over $\{W, X, Y\}$ of size at most $N^1$,
        which is again easily doable by expanding $R(X, Y)$ with the heavy $W$-value.
        Now, we have two relations over $\{W, X, Y\}$ and $\{Y, Z, W\}$ corresponding to the two bags of $\calT_2$.
        We run Yannakakis algorithm over $\calT_2$.
    \end{itemize}
    \item On the {\em heavy} $Y$-branch, we {\em cannot} afford to compute the join $R(X, Y) \Join S(Y, Z)$ over the heavy $Y$-values, thus we are {\em not} allowed to lower
    the value of $g(XYZ)$ down to 1.
    On this branch, we cannot make progress by fixing the violated submodularity
    $g(XYZ) = \infty > g(XY) + g(YZ) - g(Y) = 1$.
    Instead, we will have to make progress by utilizing something else:
    The number of heavy $Y$-values is very small, namely 1 in this example.
    Instead of lowering the value of $g(XYZ)$ as we did in the light-$Y$ branch,
    here we will lower the value of $g(Y)$ down to $\log_N 1 = 0$, which corresponds to processing the single heavy $Y$-value.
    Note that the submodularity will remain violated
    $g(XYZ)=\infty  > g(XY) + g(YZ) - g(Y) = 2$.
    However, now we can pick another submodularity to fix, for example $g(YZW) > g(Y) + g(ZW) = 1$.
    To fix this submodularity, we need to lower the value of $g(YZW)$ down to 1, which corresponds to computing a relation over $\{Y, Z, W\}$ of size at most $N^1$.
    This in turn is easily doable by expanding $T(Z, W)$ with the single heavy $Y$-value.
    After this change, we can pick another submodularity violation, for example $g(WXY) > g(WX) + g(Y) = 1$.
    This is also easily doable by expanding $U(W, X)$ with the single heavy $Y$-value.
    Now we have two relations over $\{W, X, Y\}$ and $\{Y, Z, W\}$ corresponding to the two bags of $\calT_2$, and we can run Yannakakis algorithm over $\calT_2$.
\end{itemize}
In all branches above, the \primalalgo algorithm reports the output of $Q_\square$ in time $O(N + \out)$ over the input database instance $\calD_\square$, where $\out$ is the output size.
As mentioned before, over an arbitrary input database instance, the runtime of the \primalalgo algorithm on $Q_\square$ is $O(N^{3/2 + \epsilon} + \out)$.

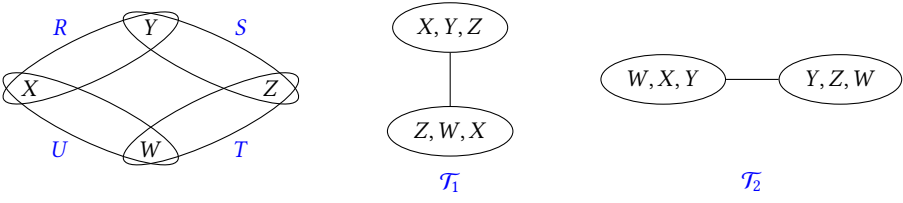
\begin{figure}[t]
    \centering
    \begin{tikzpicture}[scale = .4, every node/.style={scale=0.9}]
        \node[] at (0,0) (X) {$X$};
        \node[] at (8,0) (Z) {$Z$};
        \node[] at (4,2) (Y) {$Y$};
        \node[] at (4,-2) (W) {$W$};
        \draw[rotate around={+25:(2, +1)}, ] (2, +1) ellipse (3.2 and .75)
        node[shift={(-.45cm,+.45cm)},blue]{$R$};
        \draw[rotate around={-25:(6, +1)}, ] (6, +1) ellipse (3.2 and .75)
        node[shift={(+.45cm,+.45cm)},blue]{$S$};
        \draw[rotate around={+25:(6, -1)}, ] (6, -1) ellipse (3.2 and .75)
        node[shift={(+.45cm,-.45cm)},blue]{$T$};
        \draw[rotate around={-25:(2, -1)}, ] (2, -1) ellipse (3.2 and .75)
        node[shift={(-.45cm,-.45cm)},blue]{$U$};

        \node[draw,ellipse, right = 3 of Y] (xyz) {$X, Y, Z$};
        \node[draw,ellipse, below = 0.7 of xyz] (zwx) {$Z, W, X$};
        \draw (xyz) -- (zwx);
        \node[below = .1 of zwx] {\color{blue}$\calT_1$};

        \coordinate (TD1) at ($(xyz)!0.5!(zwx)$);
        \node[draw,ellipse, right = 2 of TD1] (wxy) {$W, X, Y$};
        \node[draw,ellipse, right = 0.7 of wxy] (yzw) {$Y, Z, W$};
        \draw (wxy) -- (yzw);
        \coordinate (TD2) at ($(wxy)!0.5!(yzw)$);
        \node[below = 1.1 of TD2] {\color{blue}$\calT_2$};
    \end{tikzpicture}
\caption{Query $Q_\square$ from Eq.~\eqref{eq:4cycle} along with the two free-connex tree decompositions.}
\label{fig:4cycle}
\end{figure}

\section{The Algorithm}
\label{sec:algorithm}

We now present the \primalalgo algorithm for evaluating a CQ $Q$ over a database instance $\calD$.

\subsection{Parameters, Inputs, and Outputs}

The $\primalalgo$ algorithm is depicted in Algorithm~\ref{algo:ppanda}.
It is parameterized by the following:
\begin{itemize}
    \item A constant $\epsilon > 0$.
    \item A CQ $Q$ (Eq.~\eqref{eq:cq}) that has variables $\bm V$ and free variables $\bm F\subseteq \bm V$.
\end{itemize}
\primalalgo takes the following inputs:
\begin{itemize}
    \item A collection of statistics $(\Delta, \bm n)$ over $\bm V$.
    \item A database instance $\calD$, that {\em includes the schema} of $Q$ (Section~\ref{sec:notation}) and satisfies the statistics $(\Delta, \bm n)$, that is, $\calD\models(\Delta, \bm n)$.
     Without loss of generality, we assume that $\calD$ is {\em signature-unique}.\footnote{
     If there are multiple relations $R(\bm X)$ over the same set of variables $\bm X$, we can
     replace them with  their intersection.}
     \revision{Moreover, we assume that $\calD$ is over the same set $\bm V$ of variables of $Q$, i.e., for every $R(\bm X)\in\calD$, we assume $\bm X \subseteq \bm V$.}
    \item {\it[Optional argument]} A set function $\bm g:2^{\bm V}\to \R_+\cup \{\infty\}$ that satisfies the following invariants:
    
    \begin{itemize}
        \item {\bf Cardinality Invariant:} For every $\bm X \subseteq \bm V$, we must have:
        \begin{align*}
            {\bm g}(\bm X) &\geq \log_N |R(\bm X)| &\text{ if } R(\bm X)\in\calD,\\
            {\bm g}(\bm X) &=\infty &\text{ otherwise.}
        \end{align*}
        \item {\bf Monotonicity Invariant:}
            $\bm g$ is monotone and satisfies ${\bm g}(\emptyset)=0$.
        \item {\bf Statistics Invariant:} $\bm g\models(\Delta, \bm n)$.
    \end{itemize}
\end{itemize}
\primalalgo is a recursive algorithm, where $\epsilon$, $Q$ and the statistics $(\Delta, \bm n)$
remain {\em the same} across all recursive calls, whereas the database instance $\calD$
and the function $\bm g$ change from one recursive call to another.
The top-level call to \primalalgo is made with a database instance $\calD_0$ that {\em matches the schema} of $Q$ (Section~\ref{sec:notation}), and this top-level call returns the answer of $Q$ over $\calD_0$, denoted $Q(\calD_0)$.
Every recursive call to $\primalalgo$ is made with a database instance~$\calD$, which is an {\em augmented} version of $\calD_0$ (Section~\ref{sec:notation}).
In particular, such an instance $\calD$ always includes the schema of $Q$, but does not necessarily match the schema of $Q$.
To every recursive call to \primalalgo over a database instance~$\calD$, we associate a CQ, $Q_\calD$, whose body is the join of {\em all} relations in the current database instance $\calD$, and whose free variables are $\bm F$, the free variables of the original query $Q$:
\begin{align}
    Q_{\calD}(\bm F) &\cd  \bigwedge_{R(\bm X)\in \calD} R(\bm X).
    \label{eq:current-cq}
\end{align}
We will show later that every recursive call to $\primalalgo$ on a database instance $\calD$
returns a set $\calA$ of tuples over the variables $\bm F$ that satisfies the following:\footnote{It might be tempting to claim that every recursive call returns $\calA := Q_{\calD}(\calD)$, but we will show later that this is not the case, and there is a fundamental reason for that.
In particular, if this was the case, then we could have solved \scq queries in time $O(N^{\subw(Q)+\epsilon})$, but counting in submodular width time remains a hard open problem~\cite{10.1145/3426865} and no prior algorithm seems to solve it~\cite{DBLP:journals/jacm/Marx13,DBLP:conf/pods/Khamis0S17,theoretics:13722,berkholz_et_al:LIPIcs.MFCS.2019.58}.}
\begin{align}
    Q_\calD(\calD) \quad\subseteq\quad \calA \quad\subseteq\quad Q(\calD_0)
    \label{eq:correctness-goal}
\end{align}
Throughout the algorithm, we use $N$ to refer to the size of the {\em original} input database instance $\calD_0$: $N\defeq |\calD_0|$.
In particular, the value of $N$ remains {\em the same} across all recursive calls to \primalalgo.

\subsection{Main Algorithm Description} 
%

In this section, we describe the steps of the \primalalgo algorithm which are depicted in Algorithm~\ref{algo:ppanda}.
As mentioned before, \primalalgo allows a set function 
$\bm g$ as an optional input argument.  In the case that this argument is not provided, the algorithm  initializes $\bm g$ as follows:
\begin{align}
g(\bm X) \defeq
    \begin{cases}
        \log_N |R(\bm X)| & \text{if } R(\bm X) \in \calD \\
        \infty & \text{otherwise}
    \end{cases}
    \label{eq:initial-g}
\end{align}
The above $\bm g$ satisfies the cardinality invariant but not necessarily the other two invariants (the monotonicity and statistics invariants).
In order to enforce these two invariants (while maintaining the first),
\primalalgo invokes a special subroutine, called \calibrate.
This subroutine replaces both $\calD$ and $\bm g$ with new database instance $\calD'$ and set function $\bm g'$ that satisfy all three invariants.

The \primalalgo algorithm then terminates if there is a tree decomposition $(T, \chi)\in \calT(Q)$ that is {\em covered by}~\footnote{\revision{Recall that the current database instance $\calD$ is always assumed to be over the same set of variables $\bm V$ as the query $Q$, i.e., for every $R(\bm X)\in\calD$, we have $\bm X \subseteq \bm V$.}} the current database instance $\calD$,
which means\footnote{This is the same as saying $(T,\chi)$  is covered by the set $\setof{\bm X}{R(\bm X)\in\calD}$, as defined in Lemma~\ref{lmm:primal-feasible}.} that for every bag $\chi(t)$ of $(T, \chi)$, there is a relation $R(\chi(t))\in\calD$.
If this is the case, \primalalgo runs the Yannakakis algorithm~\cite{DBLP:conf/vldb/Yannakakis81} over those relations $R(\chi(t))$
corresponding to the bags of $(T, \chi)$---after semijoin reducing them with all other relations to remove spurious tuples---and returns the result.

If this termination condition is not satisfied yet,
then $\primalalgo$ computes the constant $c$ from Eq.~\eqref{eq:c} in Lemma~\ref{lmm:primal-feasible}
along with the sets $\bm X, \bm Y \subseteq \bm V$ that achieve the minimum in Eq.~\eqref{eq:c}.
Later on in the analysis,
we will rely on Lemma~\ref{lmm:primal-feasible} to show that
$c \leq \subw(Q, \Delta, \bm n)$.
Let $\bm W \defeq \bm X \cap \bm Y$ and $\theta\defeq N^{g(\bm Y)-g(\bm W)}$.
Now, the algorithm partitions the relation $R(\bm Y)$ into two parts, $R^h(\bm Y)$ and $R^\ell(\bm Y)$,
called the {\em heavy} and {\em light} parts, respectively, using a threshold $\theta\times N^\epsilon$ on the degree
$\deg_R(\bm Y|\bm W)$.
In particular,
\begin{align}
    R^h(\bm Y) &\quad\defeq\quad  \left\{\bm y \in R(\bm Y)\quad\mid\quad\deg_{R(\bm Y)}(\bm Y|\bm W = \pi_{\bm W}(\bm y)) > \theta \times N^\epsilon\right\}\\
    R^\ell(\bm Y) &\quad\defeq\quad  \left\{\bm y \in R(\bm Y)\quad\mid\quad\deg_{R(\bm Y)}(\bm Y|\bm W = \pi_{\bm W}(\bm y)) \leq \theta \times N^\epsilon\right\}
\end{align}
In the light part $R^\ell(\bm Y)$, every tuple $\bm w \in \pi_{\bm W}(R^\ell(\bm Y))$ has degree at most $\theta \times N^\epsilon$.
We partition $R^\ell(\bm Y)$ further into $\ceil{N^\epsilon}$ parts
$R^{\ell, 1}(\bm Y), \ldots, R^{\ell, \ceil{N^\epsilon}}(\bm Y)$ such that
in each part $R^{\ell, i}(\bm Y)$, every tuple $\bm w \in \pi_{\bm W}(R^{\ell, i}(\bm Y))$ has degree at most $\theta$.
This can be done by a simple partitioning procedure.

Now for each light part $R^{\ell, i}(\bm Y)$, we create a database instance $\calD^{\ell, i}$
by {\em augmenting} $\calD$ with $R(\bm X) \Join R^{\ell, i}(\bm Y)$ \;(Eq.~\eqref{eq:db-augmentation}).
We also define a new set function $\bm g^{\ell}$ obtained from $\bm g$ by updating $g(\bm X \cup\bm Y)$
to $c$ \;(Eq.~\eqref{eq:function-update}).
Note that 
\begin{align}
    |R(\bm X) \Join R^{\ell, i}(\bm Y)| \leq |R(\bm X)| \cdot \theta \leq N^{g(\bm X)} \cdot N^{g(\bm Y)-g(\bm W)} = N^{c} \quad\left(\leq N^{\subw(Q, \Delta, \bm n)}\right)
\end{align}
Hence, $\bm g^\ell$ satisfies the cardinality invariant with respect to $\calD^{\ell, i}$,
and we rely on the $\calibrate$ subroutine to update $\calD^{\ell, i}$ and $\bm g^{\ell}$ to satisfy all three invariants.
We call \primalalgo recursively with the latest $\calD^{\ell, i}$ and $\bm g^{\ell}$ as the new database instance and set function, respectively, and collect the results in a set of tuples $\calA$.

Finally, we create a database instance $\calD^h$
by augmenting $\calD$ with $\pi_{\bm W}(R^h(\bm Y))$,
and we recursively call \primalalgo on $\calD^h$ {\em without} providing the optional argument $\bm g$.
Instead, we rely on Eq.~\eqref{eq:initial-g} to initialize $\bm g$ in this recursive call.
We add the result to $\calA$ and return $\calA$.

\begin{algorithm}[ht!]
    \caption{$\primalalgo_{\epsilon,Q}((\Delta, \bm n),\calD, [\bm g])$}
    \label{algo:ppanda}
    \begin{algorithmic}[1]
        \If{$\bm g$ was not provided}
            \State Initialize $\bm g$ using Eq.~\eqref{eq:initial-g}.
            \algorithmiccomment{$\bm g$ satisfies the cardinality invariant}
            \label{alg:ppanda:initialize-g}
            \State $(\calD, \bm g)\gets \calibrate((\Delta, \bm n),\calD, \bm g)$
            \algorithmiccomment{Now, $\bm g$ satisfies all three invariants}
        \EndIf
        \If{there exists $(T, \chi)\in\calT(Q)$ that is {\em covered by} $\calD$}\label{alg:ppanda:terminate}
            \State Semijoin-reduce $R(\chi(t))$ for every $t\in\nodes(T)$ with all relations in $\calD$\label{alg:ppanda:semijoin}
            \State Evaluate $Q'_\calD(\bm F) \cd \bigwedge_{t \in \nodes(T)} R(\chi(t))$ using Yannakakis algorithm and {\bf return} the result 
            \label{alg:ppanda:yannakakis}
        \Else
            \State $(\bm X, \bm Y)\quad\gets\quad \argmin_{\substack{\bm X', \bm Y'\subseteq \bm V\\g(\bm X' \cup \bm Y')>f(\bm X', \bm Y')}} f(\bm X', \bm Y')$\algorithmiccomment{$f$ is given by Eq.~\eqref{eq:f(X,Y)}}
            \label{alg:ppanda:find-XY}
            \State $c\quad\gets\quad f(\bm X, \bm Y)$\algorithmiccomment{This is the same $c$ from Eq.~\eqref{eq:c}}\label{alg:ppanda:c}
            \State $\bm W \quad\gets\quad \bm X \cap \bm Y$
            \State $\theta\quad\gets\quad N^{g(\bm Y)-g(\bm W)}$ \label{alg:ppanda:theta}
            \State $(R^h(\bm Y), R^\ell(\bm Y)) \quad\gets\quad\heavylightpartition(R(\bm Y), \bm W, \theta \times N^\epsilon)$
            \algorithmiccomment{Partition $R(\bm Y)$ into heavy and light parts based on the degree
            $\deg_{R}(\bm Y|\bm W=\bm w)$ being above or below $\theta \times N^\epsilon$}
            \label{alg:ppanda:heavylight}
            \State $(R^{\ell, 1}(\bm Y), \ldots, R^{\ell, \ceil{N^\epsilon}}(\bm Y))\quad\gets\quad
            \equaldegreepartition(R^\ell(\bm Y), \bm W, \theta)$
            \algorithmiccomment{Parition $R^\ell(\bm Y)$ into $\ceil{N^\epsilon}$ parts
            where each part $R^{\ell, i}(\bm Y)$ has degree $\deg_{R^{\ell, i}(\bm Y)}(\bm Y|\bm W=\bm w)$ at most $\theta$}
            \label{alg:ppanda:equaldegree}
            \State $\calA \quad\gets\quad \emptyset$ \algorithmiccomment{A set of tuples to collect all answers}
            \For{$i=1$ to $\ceil{N^\epsilon}$}
                \State $R^{\ell, i}(\bm X \cup \bm Y) \quad\gets\quad R(\bm X) \Join R^{\ell, i}(\bm Y)$\label{alg:ppanda:join}
                \algorithmiccomment{$|R^{\ell, i}(\bm X \cup \bm Y)| \leq N^{g(\bm X)}\cdot \theta = N^{f(\bm X, \bm Y)}=N^c$}
                \State $\calD^{\ell, i}\quad\gets\quad \calD \uplus \{R^{\ell, i}(\bm X\cup\bm Y)\}$
                \State $g^{\ell}(\bm Z)\gets \bm g[(\bm X \cup \bm Y) \to c]$\algorithmiccomment{$\bm g^\ell$ satisfies the cardinality invariant w.r.t. $\calD^{\ell, i}$}
                \label{alg:ppanda:g_ell-initialize}
                \State $(\calD^{\ell, i}, \bm g^{\ell})\gets \calibrate((\Delta, \bm n),\calD^{\ell, i}, \bm g^{\ell})$\label{alg:ppanda:g_ell}\algorithmiccomment{Now, $\bm g^\ell$ satisfies all 3 invariants w.r.t. $\calD^{\ell, i}$}
                \label{alg:ppanda:g_ell-calibrate}
                \State $\calA\quad\gets\quad \calA \cup \primalalgo_{\epsilon,Q}((\Delta, \bm n),\calD^{\ell, i}, \bm g^{\ell})$ \label{alg:ppanda:light}
            \EndFor
            \State $\calD^h\quad\gets\quad \calD \uplus \{\pi_{\bm W}(R^h(\bm Y))\}$
            \State $\calA\quad\gets\quad \calA \cup \primalalgo_{\epsilon,Q}((\Delta, \bm n),\calD^h, \mathsf{NULL})$\algorithmiccomment{We do {\em not} pass $\bm g$ here}
            \label{alg:ppanda:heavy}
            \State\Return $\calA$
        \EndIf
    \end{algorithmic}
\end{algorithm}

\subsection{The \calibrate Subroutine}

We now describe the steps of the \calibrate subroutine which are depicted in Algorithm~\ref{algo:calibrate}.
\revision{The \calibrate subroutine is mainly needed to handle {\em general} statistics $(\Delta, \bm n)$.
In particular, if all the statistics in $(\Delta, \bm n)$ are {\em cardinality upper bounds}, i.e., have the form $\log_N \deg_{R}(\bm Y|\emptyset) \leq n_{\bm Y|\emptyset}$, then \calibrate simplifies to a much simpler procedure that only enforces the monotonicity invariant, and it does so by computing projections
of every relation that is added to the current database instance.}

\calibrate takes as input a collection of statistics $(\Delta, \bm n)$, a database instance $\ov\calD$ that satisfies the statistics $(\Delta, \bm n)$, and a set function $\ov{\bm g}$ that satisfies the cardinality invariant w.r.t.~$\ov\calD$.
\calibrate returns a new database instance $\calD$ (which is an {\em augmented} version of $\ov\calD$, as defined in Section~\ref{sec:notation}) and a new set function $\bm g\leq \ov{\bm g}$ that satisfies all three invariants w.r.t.~$\calD$: the cardinality, monotonicity, and statistics invariants.

\calibrate starts by initializing its outputs $\calD$ and $\bm g$ to be copies of the inputs $\ov\calD$ and $\ov{\bm g}$.
Then,
\calibrate updates $g(\emptyset)$ to $0$ (assuming it was not $0$ already)
and augments $\calD$ with a nullary relation $R()$ consisting of a single tuple $()$, to ensure the cardinality invariant
$g(\emptyset) \geq \log_N |R()|$.
Afterwards, \calibrate repeatedly applies the following two steps until neither applies:
\begin{itemize}
    \item If the monotonicity invariant is violated, that is, there exist $\bm X \subset \bm Y \subseteq \bm V$ such that $g(\bm X) > g(\bm Y)$,
    then \calibrate updates $g(\bm X)$ to $g(\bm Y)$ and augments $\calD$ with the relation $\pi_{\bm X}(R(\bm Y))$.
    This fixes the violated monotonicity while maintaining the cardinality invariant.
    \item If the statistics invariant is violated, that is, there exists a statistics $(\bm Y|\bm X)\in\Delta$ such that
    $g(\bm X \cup \bm Y) - g(\bm X) > n_{\bm Y|\bm X}$,
    then \calibrate updates $g(\bm X\cup \bm Y)$ to become $g(\bm X) + n_{\bm Y|\bm X}$.
    In order to maintain the cardinality invariant, we need to augment $\calD$ with a relation $R(\bm X \cup \bm Y)$
    such that $|R(\bm X \cup \bm Y)| \leq N^{g(\bm X)} \cdot N^{n_{\bm Y|\bm X}}$.
    To that end, we take the {\em guard} relation $R(\bm Z)$ of the statistics $\bm Y|\bm X$ (Section~\ref{sec:notation}) and set $R(\bm X \cup \bm Y) \gets R(\bm X) \Join \pi_{\bm X\cup \bm Y}(R(\bm Z))$.
    By definition of guard relation, $\deg_{R(\bm Z)}(\bm Y|\bm X) \leq N^{n_{\bm Y|\bm X}}$.
    Moreover, by the cardinality invariant, $|R(\bm X)| \leq N^{g(\bm X)}$.
    Therefore, $|R(\bm X \cup \bm Y)|$ meets the desired bound.
\end{itemize}

\begin{algorithm}[ht!]
    \caption{$\calibrate((\Delta, \bm n), \ov\calD, \ov{\bm g})$}
    \label{algo:calibrate}
    \begin{algorithmic}[1]
        \State $(\calD, \bm g) \gets (\ov{\calD}, \ov{\bm g})$\algorithmiccomment{Initialize the outputs to the inputs}
        \State $g(\emptyset)\gets 0$\algorithmiccomment{Fix the monotonicity invariant}
        \State $R() \gets \{()\}$\algorithmiccomment{$R()$ is a nullary relation with a single tuple}
        \State $\calD \gets \calD \uplus \{R()\}$\algorithmiccomment{Re-inforce the cardinality invariant}
        \While{\true}
            \If{$\exists \bm X \subset\bm Y\subseteq \bm V$ where $g(\bm X) > g(\bm Y)$}\algorithmiccomment{If the monotonicity invariant is violated}
                \State $g(\bm X) \gets g(\bm Y)$\algorithmiccomment{Fix the monotonicity invariant}
                \State $\calD \gets \calD \uplus \{\pi_{\bm X}(R(\bm Y))\}$
                \algorithmiccomment{Re-inforce the cardinality invariant}
            \ElsIf{$\exists (\bm Y|\bm X)\in\Delta$ where $g(\bm X \cup\bm Y) - g(\bm X) > n_{\bm Y|\bm X}$}\algorithmiccomment{If the statistics invariant is violated}
                \State $g(\bm X\cup\bm Y) \gets g(\bm X) + n_{\bm Y|\bm X}$\algorithmiccomment{Fix the statistics invariant}
                \State Let $R(\bm Z)$ be the {\em guard} of $\bm Y|\bm X$
                \algorithmiccomment{$\bm Z\supseteq \bm X\cup\bm Y$ and $\log_N\deg_{R(\bm Z)}(\bm Y|\bm X) \leq n_{\bm Y|\bm X}$}
                \State $R(\bm X\cup\bm Y)\gets R(\bm X)\Join \pi_{\bm X\cup \bm Y}(R(\bm Z))$
                \algorithmiccomment{$|R(\bm X \cup \bm Y)|\leq N^{g(\bm X)} \cdot N^{n_{\bm Y|\bm X}}$}
                \State $\calD \gets \calD \uplus \{R(\bm X\cup\bm Y)\}$\algorithmiccomment{Re-inforce the cardinality invariant}
            \Else\algorithmiccomment{If no invariants are violated}
                \State \textbf{break}
            \EndIf
        \EndWhile
        \State \Return $(\calD, \bm g)$
    \end{algorithmic}
\end{algorithm}


\section{Algorithm Analysis}
\label{sec:analysis}

Our main result is the following theorem, which says that $\primalalgo$
runs in submodular width time, with an extra factor of $O(N^\epsilon)$
where $\epsilon > 0$ can be arbitrarily small.
Note that the definition of the submodular width used below {\em generalizes}
the one by Daniel Marx~\cite{DBLP:journals/jacm/Marx13} by accepting
{\em arbitrary degree constraints} (not just upper bounds on the cardinalities of relations) as well as {\em arbitrary CQs} (not just Boolean or full CQs). See Section~\ref{sec:notation}.
\begin{thm}
    Let $\epsilon > 0$ be an arbitrary constant.
    Given a CQ $Q$, a set of statistics $(\Delta, \bm n)$, and a database instance $\calD_0$ that matches the schema of $Q$ and satisfies the statistics $(\Delta, \bm n)$,
    $\primalalgo$ (Algorithm~\ref{algo:ppanda}) correctly computes $Q(\calD_0)$ in time
    \begin{align}
        O\left(N^{\subw(Q, \Delta, \bm n)+\epsilon}+ \out\right), 
        \label{eq:runtime}
    \end{align}
    in data complexity, where $N\defeq|\calD_0|$, and $\out$ is the output size.
    In particular, after a preprocessing time of $O(N^{\subw(Q, \Delta, \bm n)+\epsilon})$,
    $\primalalgo$ supports constant-delay enumeration of the output tuples.
    \label{thm:runtime}
\end{thm}

In order to prove Theorem~\ref{thm:runtime},
we first prove some basic properties of the $\calibrate$ subroutine in Section~\ref{subsec:calibrate-analysis}.
Then, we prove the correctness claim of the theorem in Section~\ref{subsec:correctness},
and finally, we prove the runtime claim in Section~\ref{subsec:runtime}.

\subsection{Analysis of the \calibrate subroutine}
\label{subsec:calibrate-analysis}
The following lemma says that, on any input database instance $\ov\calD$ and set function $\ov{\bm g}$ that satisfy the cardinality invariant, \calibrate always
terminates and outputs a database instance~$\calD$ and set function $\bm g\leq \ov{\bm g}$
that satisfy all three invariants.
In addition, suppose that we take the output $\bm g$, reduce some entry $g(\bm W)$ down to $c$,
for some value $c$, and then run \calibrate again on this modified output.
Then, the lemma says that running $\calibrate$ again cannot reduce any other entry of $\bm g$ below $c$.
All properties of \calibrate become obvious once we realize that is a {\em shortest-path algorithm} in disguise, as we explain in the proof.
In particular, the last claim in the lemma below is analogous to saying: if we compute shortest paths in a weighted
graph (with non-negative weights), and then add a new edge with weight $c$, then this newly added
edge cannot create any {\em new} shortest paths with weight less than $c$.

\begin{restatable}[Properties of \calibrate]{lemma}{CalibrateLemma}
    Let $(\Delta, \bm n)$ be a set of statistics,
    $\ov\calD$ be database instance that satisfies $(\Delta, \bm n)$,
    and $\ov{\bm g}$ be a set function  that satisfies the cardinality invariant w.r.t.~$\ov\calD$.
    Then, the \calibrate subroutine (Algorithm~\ref{algo:calibrate}) terminates and returns a database instance $\calD$ and a set function $\bm g$ that satisfy the following properties:
    \begin{enumerate}
        \item $\bm g$ satisfies the cardinality, monotonicity and statistics invariants w.r.t.~$\calD$. \label{lmm:calibrate:1}
        \item $\bm g \leq \ov{\bm g}$. \label{lmm:calibrate:2}
        \item Suppose that there exists a set function $\tl{\bm g}$
        that satisfies all three invariants w.r.t.~$\ov\calD$ and
        the input~$\ov{\bm g}$ satisfies:
        $\ov{\bm g} = \tl{\bm g}[\bm W \to c]$ for some $\bm W \subseteq \bm V$, $c \in \R_+$
        where $c < \tl{g}(\bm W)$.
        Then, the output~$\bm g$ satisfies: $g(\bm Z) \geq c$ for every $\bm Z\subseteq \bm V$ where $g(\bm Z) < \ov g(\bm Z)$. \label{lmm:calibrate:3}
    \end{enumerate}
    \label{lmm:calibrate}
\end{restatable}
\begin{proof}
    First, we prove that \calibrate always terminates.
    To that end, we show that \calibrate is a shortest-path algorithm in disguise on
    a certain weighted directed graph.
    In particular, consider the following weighted directed graph $\ov G$
    whose vertices are all subsets of $\bm V$ and has four kinds of weighted edges:
    \begin{itemize}
        \item For every $\bm X \subseteq \bm V$, we have an edge $\emptyset \to \bm X$ with weight $\ov g(\bm X)$.
        \item We have an edge $\emptyset \to \emptyset$ with weight 0.
        \item For every $\bm X \subset \bm Y \subseteq \bm V$,
        we have an edge $\bm Y \to \bm X$ with weight 0.
        \item For every statistics $(\bm Y|\bm X)$ in $\Delta$,
        we have an edge $\bm X \to (\bm X \cup \bm Y)$ with weight $n_{\bm Y|\bm X}$.
    \end{itemize}
    It is straightforward to verify that \calibrate
    is computing shortest paths from $\emptyset$ to every other vertex in $\ov G$.
    In particular, each step of \calibrate is considering some edge $\bm X \to \bm Y$
    and trying to use the current shortest path to $\bm X$ to update the shortest path to $\bm Y$.
    Because all weights in the graph are non-negative, for a fixed $\bm Y \subseteq \bm V$,
    the total number of times the algorithm can lower the value of $g(\bm Y)$
    is upper bounded by the number of simple paths from $\emptyset$ to $\bm Y$ in $\ov G$.
    Hence, the algorithm terminates.

    As long as the algorithm terminates, the output function $\bm g$ must satisfy
    the all three invariants, since the algorithm keeps going as long as
    any invariant is violated.
    Hence, Claim~(\ref{lmm:calibrate:1}) holds.
    Claim~(\ref{lmm:calibrate:2}) is obvious.

    Finally, we prove Claim~(\ref{lmm:calibrate:3}).
    Let $\tl{G}$ be the weighted directed graph defined similarly to $\ov G$ but
    using $\tl{\bm g}$ instead of $\ov{\bm g}$.
    In the language of shortest paths, Claim~(\ref{lmm:calibrate:3}) says that
    $\tl G$ is {\em closed} under shortest paths, meaning that for every vertex $\bm Y$ in $\tl G$,
    there is a shortest path from $\emptyset$ to $\bm Y$ in $\tl G$ consisting of a {\em single edge}.
    The claim also says that $\ov G$ is obtained from $\tl G$ by introducing a new edge
    $\emptyset \to \bm W$ with weight~$c$.
    Because all weights are non-negative, all {\em new} shortest paths in $\ov G$ that were not present in $\tl G$ (i.e., the shortest paths that use the newly added edge) must
    have weight at least $c$.
\end{proof}

\begin{proposition}[Runtime of \calibrate]
    \label{prop:calibrate-runtime}
    The $\calibrate$ algorithm (Algorithm~\ref{algo:calibrate}) runs in linear time in the size of the input database instance $\ov \calD$, in data complexity.
\end{proposition}

\subsection{Correctness Proof of \primalalgo}
\label{subsec:correctness}
As mentioned before, $\primalalgo$ is recursive where the top-level call is made with a database instance~$\calD_0$ (of size $N$), and each recursive call is made with an {\em augmented} database instance $\calD$ (of size potentially larger than $N$).
It is not immediately clear that the algorithm terminates, i.e., that the recursion depth is finite.
However, assuming that it does, we will prove in this section that it is correct, i.e., the top-level call to \primalalgo returns $Q(\calD_0)$.
In the next section, we will prove a finite bound on the recursion depth, hence proving termination.

In order to prove that \primalalgo is correct, we will inductively prove the following claim:
\begin{claim}
    \label{clm:correctness}
    Every recursive call to \primalalgo on a database instance $\calD$
    returns a set of tuples $\calA \subseteq \dom^{\bm F}$ that satisfies Eq.~\eqref{eq:correctness-goal}.
\end{claim}
\begin{proof}[Proof of Claim~\ref{clm:correctness}]
    As a base case, consider a recursive call that terminates in line~\ref{alg:ppanda:yannakakis}
    of Algorithm~\ref{algo:ppanda}.
    The current database instance $\calD$ is always an augmented version of the original input  instance $\calD_0$.
    In addition, for every relation $R(\bm X)$ in $\calD_0$, there is a bag $R(\chi(t))$
    for some $t \in \nodes(T)$ where $\bm X \subseteq \chi(t)$ and $R(\chi(t))$ was semijoin-reduced
    with $R(\bm X)$ in line~\ref{alg:ppanda:semijoin}.
    This proves that $\calA \defeq Q'_{\calD}(\calD) \subseteq Q(\calD_0)$.
    Moreover, $Q_\calD(\calD) \subseteq Q'_{\calD}(\calD)$ by definition.\footnote{Note that in this step, we can only compute $Q'_{\calD}(\calD)$ but {\em not} necessarily $Q_\calD(\calD)$ because there is no guarantee that for every relation $R(\bm X)$ in the current database instance $\calD$, there will be a bag $R(\chi(t))$ for some $t \in \nodes(T)$ where $\bm X \subseteq \chi(t)$.
    This is because $(T,\chi)$ is a tree decomposition of the {\em original query} $Q$ but not necessarily of $Q_\calD$.
    This seems to be a fundamental limitation of similar algorithms, where they cannot be extended to counting (i.e., solving \scq) in submodular width time $O(N^{\subw(Q)})$.
    The same limitation applies to the original algorithm proposed by Daniel Marx~\cite{DBLP:journals/jacm/Marx13} as well as all subsequent algorithms~\cite{DBLP:conf/pods/Khamis0S17,theoretics:13722,berkholz_et_al:LIPIcs.MFCS.2019.58}.}

    For the inductive step, note that we are partitioning the relation $R(\bm Y)$
    into (disjoint) parts in both lines~\ref{alg:ppanda:heavylight} and~\ref{alg:ppanda:equaldegree}.
    Afterwards, for each part, we augment the database instance $\calD$ with that part and recursively call \primalalgo.
    By induction, it is straightforward to see that $\calA \subseteq Q(\calD_0)$
    and at the same time $\calA \supseteq Q_\calD(\calD)$.
\end{proof}

\subsection{Runtime Analysis of \primalalgo}
\label{subsec:runtime}
In this section, we prove that $\primalalgo$ runs in time
$O(N^{\subw(Q, \Delta, \bm n)+\epsilon}+$ $\out)$ in data complexity, and also supports constant-delay enumeration after a pre-processing time of $O(N^{\subw(Q, \Delta, \bm n)+\epsilon})$.
Along with the previous section, this will complete the proof of Theorem~\ref{thm:runtime}.

    The execution of $\primalalgo$ can be represented as a recursion tree, where each node is a call to $\primalalgo$.
    Each internal node has $\ceil{N^\epsilon} + 1$ children, namely $\ceil{N^\epsilon}$ children
    corresponding to light parts in line~\ref{alg:ppanda:light} and one child
    corresponding to the heavy part in line~\ref{alg:ppanda:heavy}.
    An edge from a parent to a child is called {\em light} if it corresponds to a recursive call on a light partition in line~\ref{alg:ppanda:light},
    and {\em heavy} otherwise.
    A path from an ancestor node to a descendant node is called a {\em light path} if it consists of only light edges. Our first target is to prove the following claim:

    \begin{claim}
        \label{clm:light-path}
        A light path can have length at most $2^{|\bm V|}=O(1)$.
    \end{claim}
    In order to prove Claim~\ref{clm:light-path}, we will prove a series of claims.
    To that end, we introduce some notation.
    Let $f^\ell, c^\ell, h^\ell, \calI^\ell$ be defined as in Eq.~\eqref{eq:f(X,Y)}-\eqref{eq:I} but, instead of $\bm g$, they are with respect to $\bm g^\ell$ {\em after} calibration is complete in line~\ref{alg:ppanda:g_ell}.

    \begin{claim}
        \label{clm:calibration}
        For every $\bm Z \subseteq \bm V$, $g^\ell(\bm Z) \leq g(\bm Z)$.
        Moreover, if $g^\ell(\bm Z) < g(\bm Z)$, then $g^\ell(\bm Z)\geq c$.
    \end{claim}
    \begin{proof}[Proof of Claim~\ref{clm:calibration}]
    Claim~\ref{clm:calibration} follows immediately from Lemma~\ref{lmm:calibrate}.
    In particular, the lemma says that $\bm g^\ell \leq \bm g$.
    Moreover, $\bm g$ satisfies all three invariants w.r.t. $\calD$
    and $g^{\ell}(\bm Z)\defeq \bm g[(\bm X \cup \bm Y) \to c]$, per line~\ref{alg:ppanda:g_ell-initialize} of Algorithm~\ref{algo:ppanda}.
    Thanks to the last claim of Lemma~\ref{lmm:calibrate}, the \calibrate call in line~\ref{alg:ppanda:g_ell-calibrate}
    cannot reduce any $g^{\ell}(\bm Z)$ value below $c$.
    \end{proof}

    We now use Claim~\ref{clm:calibration} to prove the following claim.
    \begin{claim}
        \label{clm:c-ell}
        $c^\ell \geq c$ and $\calI^\ell \supset \calI$.
    \end{claim}
    \begin{proof}[Proof of Claim~\ref{clm:c-ell}]
        First, we prove that $c^\ell \geq c$.
        To that end, we show that for every pair $\bm X', \bm Y' \subseteq \bm V$
        where $\bm g^\ell$ violates submodularity, i.e., where $g^\ell(\bm X' \cup \bm Y') > f^\ell(\bm X', \bm Y')$, we must have $f^\ell(\bm X', \bm Y') \geq c$.
        We consider two cases:
        \begin{itemize}
        \item Suppose that the pair $\bm X', \bm Y'$ above used to satisfy submodularity
        in $\bm g$, i.e., $g(\bm X' \cup \bm Y') \leq f(\bm X', \bm Y')$.
        Since the pair violates submodularity in $\bm g^\ell$,
        it must be the case that $g^\ell(\bm X') < g(\bm X')$ (or symmetrically $g^\ell(\bm Y') < g(\bm Y')$). In this case, $g^\ell(\bm X') \geq c$ by Claim~\ref{clm:calibration}, hence
        $f^\ell(\bm X', \bm Y') \defeq g^\ell(\bm X') + g^\ell(\bm Y') - g^\ell(\bm X' \cap \bm Y') \geq g^\ell(\bm X') \geq c$ where the first inequality holds by monotonicity of $\bm g^\ell$.

        \item Now suppose that the pair $\bm X', \bm Y'$ already violated submodularity in $\bm g$, i.e., $g(\bm X'\cup \bm Y') > f(\bm X', \bm Y')\geq c$ (where the last inequality follows by definition of $c$).
        We recognize the following cases:
        \begin{itemize}
            \item If $g^\ell(\bm X')<g(\bm X')$, then $g^\ell(\bm X')\geq c$ by Claim~\ref{clm:calibration} and $f^\ell(\bm X', \bm Y') \defeq g^\ell(\bm X') + g^\ell(\bm Y') - g^\ell(\bm X'\cap \bm Y') \geq g^\ell(\bm X') \geq c$ by monotonicity of $\bm g^\ell$.
            \item If $g^\ell(\bm Y')<g(\bm Y')$, then this is symmetric to the above case.
            \item If $g^\ell(\bm X')=g(\bm X')$ and $g^\ell(\bm Y') = g(\bm Y')$, then
            $f^\ell(\bm X', \bm Y') = g(\bm X') + g(\bm Y') - g^{\ell}(\bm X' \cap \bm Y') \geq
            g(\bm X') + g(\bm Y') - g(\bm X' \cap \bm Y') = f(\bm X', \bm Y') \geq c$.
        \end{itemize}
        \end{itemize}
        In all above cases, we have $f^\ell(\bm X', \bm Y') \geq c$, hence $c^\ell \geq c$.

        The fact that $\calI^\ell\supseteq \calI$ follows from $c^\ell\geq c$ and $g^\ell(\bm Z)\leq g(\bm Z)$ for all $\bm Z$. Moreover, $\calI^\ell$ contains $\bm X\cup\bm Y$ from line~\ref{alg:ppanda:find-XY},
        while $\calI$ doesn't, thus $\calI^\ell\supset \calI$.
    \end{proof}
    Now, we are ready to prove Claim~\ref{clm:light-path}.
    \begin{proof}[Proof of Claim~\ref{clm:light-path}]
   By Claim~\ref{clm:c-ell},
    the size of $\calI$ increases by at least one with every recursive call to the light branch.
    Moreover, this size cannot exceed $2^{|\bm V|}$.
    \end{proof}
%
    
    Claim~\ref{clm:light-path} proves an upper bound on the number of consecutive light edges
    on any path from the root to a leaf in the recursion tree.
    In contrast, the following claim proves an upper bound on the number of heavy edges
    on any path from the root to a leaf.
    The two claims together imply an upper bound on the depth of any leaf.
    \begin{claim}
        \label{clm:heavy-edges}
        Any path from the root to a leaf can have at most $O(1)$ heavy edges.
    \end{claim}
    \begin{proof}[Proof of Claim~\ref{clm:heavy-edges}]
        Let $\Gamma$ be a constant that is larger than the full join of all relation instances in $\calD_0$, e.g.,
        $\Gamma\defeq N^{|\bm V|+1}$.
        Define the following potential function:
        \begin{align*}
            \varphi(\calD) &\defeq \sum_{R(\bm X)\in\calD} \log_N |R(\bm X)|
            +\sum_{\substack{\bm X\subseteq \bm V\\R(\bm X)\notin\calD}} \log_N \Gamma
        \end{align*}
        Clearly, $\varphi(\calD) = O(1)$ in data complexity.
        Consider a path from the root to a leaf.
        We argue that between every two consecutive heavy edges (with an arbitrary number of light edges in between), the potential decreases by at least $\epsilon$.
        In particular, let $\ov{\calD}^h$ and ${\calD}^{h}$ be the input database instances to two consecutive heavy calls to $\primalalgo$ in line~\ref{alg:ppanda:heavy} (with potentially many light calls in line~\ref{alg:ppanda:light} in between),
        and let $\ov{\bm g}$ and $\bm g$ be the two corresponding set functions.
        We argue that $\varphi({\calD}^{h}) \leq \varphi(\ov{\calD}^{h}) - \epsilon$.
        The first heavy call will initialize $\ov{\bm g}$ in line~\ref{alg:ppanda:initialize-g}
        to satisfy the cardinality invariant. From that point until the second heavy call, the cardinality invariant
        will always be maintained and $g(\bm X)$ can never increase for any $\bm X\subseteq \bm V$.
        At the time we reach the second heavy call in line~\ref{alg:ppanda:heavy}, we have
        \begin{align*}
            \log_N |R^{h}(\bm X\cap\bm Y)| &\leq \log_N \frac{|R(\bm Y)|}{\theta \times {N^\epsilon}}\\
            &= \log_N |R(\bm Y)| - \log_N \theta - \epsilon \\
            &\leq g(\bm Y) - \log_N \theta - \epsilon &\text{(by the cardinality invariant)}\\
            &= g(\bm X\cap\bm Y) - \epsilon &\text{(by definition of $\theta$ in line~\ref{alg:ppanda:theta})}\\
            &\leq \ov g(\bm X\cap\bm Y) - \epsilon &\text{($\bm g$ never increases between heavy calls)}\\
            &= \log_N |\ov{R}^h(\bm X\cap \bm Y)| - \epsilon &\text{(by initialization of $\ov{\bm g}$ in line~\ref{alg:ppanda:initialize-g})}
        \end{align*}
        This proves $\varphi({\calD}^{h}) \leq \varphi(\ov{\calD}^{h}) - \epsilon$.
        Moreover, since $\varphi(\calD) = O(1)$, the number of heavy edges in any path from the root to a leaf is $O(1)$.
    \end{proof}

    Claims~\ref{clm:light-path} and~\ref{clm:heavy-edges} together imply the following claim:
    \begin{claim}
        \label{clm:depth}
        The depth of the recursion tree of $\primalalgo$ is $O(1)$ in data complexity.
    \end{claim}

    Claim~\ref{clm:depth} along with the fact that each internal (i.e., non-leaf) node has $\ceil{N^\epsilon}+1$ children imply that the size of the recursion tree is $O(N^{\epsilon'})$ for some constant $\epsilon'>0$.
    Next, we argue that the time spent at each internal node is $O(N^{\subw(Q, \Delta, \bm n)+\epsilon})$. To that end, we need the following claim:
    \begin{claim}
        \label{clm:c-subw}
        In line~\ref{alg:ppanda:c}, $c \leq \subw(Q, \Delta, \bm n)$.
    \end{claim}
    \begin{proof}[Proof of Claim~\ref{clm:c-subw}]
        First, we show that $c$ in line~\ref{alg:ppanda:c} cannot be $\infty$. In particular,
        if $c=\infty$, then $\bm g$ is already a polymatroid, hence
        $\bm g <\infty$.
        By the cardinality invariant, $\calD$ must contain a finite relation $R(\bm X)$
        for every $\bm X \subseteq \bm V$, but then $\calD$ must cover every tree decomposition in $\calT(Q)$,
        contradicting the failure of the condition in line~\ref{alg:ppanda:terminate}.

    Assuming $c< \infty$, Claim~\ref{clm:c-subw} follows from Lemma~\ref{lmm:primal-feasible}.
    In particular, 
    the condition in line~\ref{alg:ppanda:terminate} ensures that there is no tree decomposition in $\calT(Q)$ that is covered by
    \begin{align*}
        \calI'\defeq\setof{\bm X \subseteq \bm V}{R(\bm X) \in \calD} \supseteq
        \setof{\bm X \subseteq \bm V}{g(\bm X) < \infty} \supseteq
        \setof{\bm X \subseteq \bm V}{g(\bm X) \leq c} = \calI
    \end{align*}
    The first $\supseteq$ above follows from the cardinality invariant, while the second follows from $c < \infty$.
    The statistics invariant ensures that $\bm g\models(\Delta, \bm n)$, hence, Lemma~\ref{lmm:primal-feasible} applies.
    \end{proof}

    Thanks to Claim~\ref{clm:c-subw}, the time needed to compute each join $R^{\ell, i}(\bm X \cup \bm Y)$ in line~\ref{alg:ppanda:join} is within
    \begin{align}
        |R^{\ell, i}(\bm X \cup \bm Y)| \leq |R(\bm X)|\cdot \deg_{R^{\ell, i}}(\bm Y|\bm X\cap \bm Y) \leq N^{g(\bm X)} \times \theta=
        N^{f(\bm X, \bm Y)} =N^c \leq N^{\subw(Q, \Delta, \bm n)}
        \label{eq:join-size}
    \end{align}
    Therefore, the total time to compute the joins for $i$ from $1$ to $\ceil{N^\epsilon}$ is
    $O(N^{\subw(Q, \Delta, \bm n)+\epsilon})$.
    Moreover, by Proposition~\ref{prop:calibrate-runtime}, the time spent in \calibrate is linear
    in the size of the current database $\calD$, which cannot exceed $O(N^{\subw(Q, \Delta, \bm n)})$, by Eq.~\eqref{eq:join-size}.
    This proves that the time spent at each internal node is $O(N^{\subw(Q, \Delta, \bm n)+\epsilon})$, hence the time spent at {\em all} internal nodes combined is $O(N^{\subw(Q, \Delta, \bm n)+\epsilon+\epsilon'})$.

    Finally, we analyze the time needed to report the output at leaf nodes.
    The recursion tree can have $O(N^{\epsilon'})$ leaves, and each output tuple can occur in multiple leaves. However, the total number of tree decompositions is constant in data complexity,
    thus we can group together leaves that use the same tree decomposition in line~\ref{alg:ppanda:terminate} (by taking the union of their bags).
    Thanks to this grouping, each output tuple can still occur in multiple groups but the total number of groups is constant, thus allowing us to remove duplicates with extra $O(1)$ overhead and report the output in $O(\out)$ time.
    This proves that the overall runtime of $\primalalgo$ is $O(N^{\subw(Q, \Delta, \bm n)+\epsilon+\epsilon'}+ \out)$, as desired.
    
    \paragraph*{Constant-delay enumeration (CDE)}
    In order to support CDE after a pre-processing time of $O(N^{\subw(Q, \Delta, \bm n)+\epsilon+\epsilon'})$,
    we replace running Yannakakis algorithm in line~\ref{alg:ppanda:yannakakis}
    with the enumeration algorithm from~\cite{DBLP:conf/csl/BaganDG07}.
    (Recall from Section~\ref{sec:notation} that all tree decompositions in $\calT(Q)$ are {\em free-connex}.)
    As mentioned above, the same output tuple can still occur in $O(1)$ leaves,
    hence can have $O(1)$ duplicates.
    However, we can use the {\em Cheater's Lemma}~\cite{10.1145/3450263} as a standard technique
    to support CDE of the union of a constant number of sets,
    each of which supports CDE. This completes the proof of all claims made in Theorem~\ref{thm:runtime}.
\section{Conclusion and Open Problems}
\label{sec:conclusion}

We introduced an algorithm, called \primalalgo, for evaluating a CQ $Q$ in submodular width time,
with an extra overhead of $O(N^\epsilon)$ where $\epsilon>0$ can be arbitrarily small.
The algorithm handles arbitrary~{\em statistics}, which are degree constraints, on the input database instance.
These statistics can significantly lower the value of the submodular width compared to the original definition~\cite{DBLP:journals/jacm/Marx13}.
The submodular width is defined via a maximization problem which naturally has equivalent primal and dual formulations.
Correspondingly, we classify query evaluation algorithms that target the submodular width into two classes: (a) {\em primal algorithms} which operate in the primal space (the space of polymatroids) of this maximization problem
and include the algorithms from~\cite{DBLP:journals/jacm/Marx13, berkholz_et_al:LIPIcs.MFCS.2019.58} as well as our new algorithm; and (b) {\em dual algorithms} which operate in the dual space
(the space of Shannon inequalities) and include PANDA~\cite{DBLP:conf/pods/Khamis0S17,theoretics:13722}.
Because it is a dual one, PANDA is quite involved and requires transforming a CQ
into disjunctive datalog rules (DDRs) as well as using Shannon entropy.
In contrast, our new primal algorithm is simpler, more direct, and bypasses the need for DDRs and Shannon entropy.

We conclude with a couple of major open problems in this area.

\paragraph*{\bf Solving \scqs in submodular width time.}
It might be tempting at first to try to extend our algorithm to solve counting conjunctive queries (\scqs) in time $O(N^{\subw(Q)+\epsilon})$. This is because the algorithm recursively partitions the database
instance into {\em disjoint} parts (lines~\ref{alg:ppanda:heavylight} and~\ref{alg:ppanda:equaldegree} of Algorithm~\ref{algo:ppanda}).
However, as explained earlier, there is no guarantee that the tree decomposition found at the end of the recursion (line~\ref{alg:ppanda:terminate}) will actually encompass all intermediate relations that the algorithm created.
This seems to be a common weakness of {\em all} prior algorithms~\cite{DBLP:journals/jacm/Marx13, berkholz_et_al:LIPIcs.MFCS.2019.58,DBLP:conf/pods/Khamis0S17,theoretics:13722}, which also fall
short of solving \scqs in submodular width time.
To circumvent this issue, a weaker version of the submodular width, called the {\em sharp-submodular width}, that allows for counting, was introduced in~\cite{10.1145/3426865}.
However, it was shown recently
that this new width is not the best possible: There are \scqs where counting can be done faster
than the sharp-submodular width time~\cite{10.1145/3717823.3718141}.
This leaves open the question of whether we can actually solve \scqs in submodular width time,
and if not, what is the best possible complexity measure for \scqs.

\paragraph*{\bf Solving CQs in {\em entropic width} time.}
The current submodular width definition from Eq.~\eqref{eq:subw} takes a maximum over all polymatroids. It is possible to derive a tighter width notion by maximizing over only the {\em entropic functions}.
This tighter notion has been referred to as the {\em entropic width} in~\cite{DBLP:conf/pods/Khamis0S17}, and we denote it here by $\entw(Q)$.
Another way to put it is that the submodular width only considers Shannon inequalities, whereas
the entropic width additionally includes {\em non-Shannon inequalities}~\cite{DBLP:journals/tit/ZhangY97,zhang1998characterization}.
Recent work has shown a lower bound\footnote{For consistency, we assume here that $\entw(Q)$ is defined with log base $N$, similar to Remark~\ref{rmk:log-base}.} of $\Omega(N^{(1-\epsilon)\cdot\entw(Q)})$ for every $\epsilon > 0$ on the size of the {\em semiring circuit} needed to represent the output of a CQ~\cite{10.1145/3651588}.
However, no algorithm is known to answer a CQ in time written in terms of $O(N^{\entw(Q)})$.
Achieving this bound remains an open problem~\cite{kolaitis_et_al:DagRep.12.7.180}.

\begin{acks}
H. Chen acknowledges the support of EPSRC grants EP/V039318/1 and EP/X03190X/1.
\end{acks}

\bibliographystyle{ACM-Reference-Format}
\bibliography{bib,hubiebib}


\end{document}